%% file: p.tex
\documentclass[review=false,onefignum,onetabnum]{siamonline190516}%

\input{p_shared}

\ifpdf
\hypersetup{
  pdftitle={On $\mu$-Symmetric Polynomials},
  pdfauthor={Jing Yang and Chee Yap}
}
\fi

\begin{document}

\maketitle

\begin{abstract}
	In this paper, we
	study functions of the roots of a univariate
	polynomial in which the roots have a given
	multiplicity structure $\bfmu$.
	Traditionally, root functions are studied via
	the theory of symmetric polynomials; we extend this
	theory to $\bfmu$-symmetric polynomials.
	We were motivated by a conjecture from Becker et al.~(ISSAC 2016)
	about the $\bfmu$-symmetry of a particular
	root function $\dplus(\bfmu)$, called D-plus.
	To investigate this conjecture,
    it was desirable to have fast algorithms for checking
    if a given root function is $\bfmu$-symmetric.
	We designed three such algorithms:
        one based on Gr\"{o}bner bases,
        another based on preprocessing and reduction,
        and the third based on solving linear equations.
    We implemented them in Maple and
	experiments show that the latter two
        algorithms are significantly faster than the first.
\end{abstract}



\ignore{
\section{Introduction}
The introduction introduces the context and summarizes the
manuscript. It is importantly to clearly state the contributions of
this piece of work. The next two paragraphs are text filler,
generated by the \texttt{lipsum} package.

\lipsum[2-3]

The paper is organized as follows. Our main results are in
\cref{sec:main}, our new algorithm is in \cref{sec:alg}, experimental
results are in \cref{sec:experiments}, and the conclusions follow in
\cref{sec:conclusions}.

\section{Main results}
\label{sec:main}

We interleave text filler with some example theorems and theorem-like
items.

\lipsum[4]

Here we state our main result as \cref{thm:bigthm}; the proof is
deferred to \cref{sec:proof}.

\begin{theorem}[$LDL^T$ Factorization \cite{GoVa13}]\label{thm:bigthm}
  If $A \in \mathbb{R}^{n \times n}$ is symmetric and the principal
  submatrix $A(1:k,1:k)$ is nonsingular for $k=1:n-1$, then there
  exists a unit lower triangular matrix $L$ and a diagonal matrix
  \begin{displaymath}
    D = \diag(d_1,\dots,d_n)
  \end{displaymath}
  such that $A=LDL^T$. The factorization is unique.
\end{theorem}

\lipsum[6]

\begin{theorem}[Mean Value Theorem]\label{thm:mvt}
  Suppose $f$ is a function that is continuous on the closed interval
  $[a,b]$.  and differentiable on the open interval $(a,b)$.
  Then there exists a number $c$ such that $a < c < b$ and
  \begin{displaymath}
    f'(c) = \frac{f(b)-f(a)}{b-a}.
  \end{displaymath}
  In other words,
  \begin{displaymath}
    f(b)-f(a) = f'(c)(b-a).
  \end{displaymath}
\end{theorem}

Observe that \cref{thm:bigthm,thm:mvt,cor:a} correctly mix references
to multiple labels.

\begin{corollary}\label{cor:a}
  Let $f(x)$ be continuous and differentiable everywhere. If $f(x)$
  has at least two roots, then $f'(x)$ must have at least one root.
\end{corollary}
\begin{proof}
  Let $a$ and $b$ be two distinct roots of $f$.
  By \cref{thm:mvt}, there exists a number $c$ such that
  \begin{displaymath}
    f'(c) = \frac{f(b)-f(a)}{b-a} = \frac{0-0}{b-a} = 0.
  \end{displaymath}
\end{proof}

Note that it may require two \LaTeX\ compilations for the proof marks
to show.

Display matrices can be rendered using environments from \texttt{amsmath}:
\begin{equation}\label{eq:matrices}
S=\begin{bmatrix}1&0\\0&0\end{bmatrix}
\quad\text{and}\quad
C=\begin{pmatrix}1&1&0\\1&1&0\\0&0&0\end{pmatrix}.
\end{equation}
Equation \cref{eq:matrices} shows some example matrices.

We calculate the Fr\'{e}chet derivative of $F$ as follows:
\begin{subequations}
\begin{align}
  F'(U,V)(H,K)
  &= \langle R(U,V),H\Sigma V^{T} + U\Sigma K^{T} -
  P(H\Sigma V^{T} + U\Sigma K^{T})\rangle \label{eq:aa} \\
  &= \langle R(U,V),H\Sigma V^{T} + U\Sigma K^{T}\rangle
  \nonumber \\
  &= \langle R(U,V)V\Sigma^{T},H\rangle +
  \langle \Sigma^{T}U^{T}R(U,V),K^{T}\rangle. \label{eq:bb}
\end{align}
\end{subequations}
\Cref{eq:aa} is the first line, and \cref{eq:bb} is the last line.

\section{Algorithm}
\label{sec:alg}

\lipsum[40]

Our analysis leads to the algorithm in \cref{alg:buildtree}.

\begin{algorithm}
\caption{Build tree}
\label{alg:buildtree}
\begin{algorithmic}
\STATE{Define $P:=T:=\{ \{1\},\ldots,\{d\}$\}}
\WHILE{$\#P > 1$}
\STATE{Choose $C^\prime\in\mathcal{C}_p(P)$ with $C^\prime := \operatorname{argmin}_{C\in\mathcal{C}_p(P)} \varrho(C)$}
\STATE{Find an optimal partition tree $T_{C^\prime}$ }
\STATE{Update $P := (P{\setminus} C^\prime) \cup \{ \bigcup_{t\in C^\prime} t \}$}
\STATE{Update $T := T \cup \{ \bigcup_{t\in\tau} t : \tau\in T_{C^\prime}{\setminus} \mathcal{L}(T_{C^\prime})\}$}
\ENDWHILE
\RETURN $T$
\end{algorithmic}
\end{algorithm}

\lipsum[41]

\section{Experimental results}
\label{sec:experiments}

\lipsum[50]

\Cref{fig:testfig} shows some example results. Additional results are
available in the supplement in \cref{tab:foo}.

\begin{figure}[htbp]
  \centering
  \label{fig:a}\includegraphics{lexample_fig1}
  \caption{Example figure using external image files.}
  \label{fig:testfig}
\end{figure}

\lipsum[51]

\section{Discussion of \texorpdfstring{{\boldmath$Z=X \cup Y$}}{Z = X union Y}}

\lipsum[76]

\section{Conclusions}
\label{sec:conclusions}

Some conclusions here.

\appendix
\section{An example appendix}
\lipsum[71]

\begin{lemma}
Test Lemma.
\end{lemma}

}%

\input{p-fullpaper-musym}

\section*{Acknowledgments}
The authors would like to thank Professor Dongming Wang for his
generous support and the use of the resources of the SMS International.
This work is done during Chee's sabbatical year in China with
the generous support of GXUN, Beihang University and Chinese Academy
of Sciences in Beijing.

\bibliographystyle{siamplain}

\input{p-save.bbl}
\end{document}

%% file: p_shared.tex
\usepackage{lipsum}
\usepackage{amsfonts}
\usepackage{graphicx}
\usepackage{epstopdf}
\usepackage{algorithmic}
\ifpdf
  \DeclareGraphicsExtensions{.eps,.pdf,.png,.jpg}
\else
  \DeclareGraphicsExtensions{.eps}
\fi

	\usepackage{myMac-SIAM}	

\usepackage{multicol}
\usepackage{multirow}

    \newcommand{\oldNew}[2]{\color{red} NOTHING SHOWN } 
      \renewcommand{\oldNew}[2]{{\color{red} OLD: #1}
	  {\color{blue} NEW: #2 }}
      \renewcommand{\oldNew}[2]{#2}
    \newcommand{\revise}[1]{{\color{red} #1}}
    \newcommand{\revised}[1]{{#1}}
	\newcommand{\dstar}{D\sstar}
	\newcommand{\dplus}{D^+}

	\newcommand{\bfrho}{{\bs{\rho}}}		
	\newcommand{\sym}{_{\tt{sym}}}		%
	\newcommand{\lift}[1]{\Big(#1 \Big)^{\wedge}}		%
    \newcommand{\ooF}{\mathring{F}}         %
    \newcommand{\e}{e}      %
    \newcommand{\ole}{\ol{\e}}      %
	\newcommand{\Supp}{{\tt Supp}}	%
	\newcommand{\coef}{{\tt Coef}}		%
	\newcommand{\coeffs}{{\tt Coeffs}}	%
	\newcommand{\lm}{{\tt Lm}}		%
	\newcommand{\lt}{{\tt Lt}}		%
	\newcommand{\lco}{{\tt Lc}}		%
	\newcommand{\wdeg}{\omega\text{-deg}}		%

	\newcommand{\ggist}{{\tt G\hyp gist}}		%
	\newcommand{\lsgist}{{\tt LS\hyp gist}}		%
	\newcommand{\crgist}{{\tt CR\hyp gist}}		%
	\newcommand{\eGgist}{{\tt e\hyp Ggist}}		%
	\newcommand{\eCRgist}{{\tt e\hyp CRgist}}		%
	\newcommand{\eLSgist}{{\tt e\hyp LSgist}}		%
	\newcommand{\mCRgist}{{\tt m\hyp CRgist}}		%
	\newcommand{\mLSgist}{{\tt m\hyp LSgist}}		%
	\newcommand{\cGgist}{{\tt c\hyp Ggist}}		%
	\newcommand{\cCRgist}{{\tt c\hyp CRgist}}		%
	\newcommand{\cLSgist}{{\tt c\hyp LSgist}}		%
	\newcommand{\pGgist}{{\tt p\hyp Ggist}}		%
	\newcommand{\pCRgist}{{\tt p\hyp CRgist}}		%
	\newcommand{\pLSgist}{{\tt p\hyp LSgist}}		%
    \newcommand{\canonize}{{\tt canonize}}          %
    \newcommand{\reduce}{{\tt reduce}}              %
    
    \newcommand{\reduceStep}{{\tt reduceStep}}              %
    \newcommand{\nreduce}{{\tt nreduce}}
    \newcommand{\rstar}{\color{red}{$^*$}}
	\newcommand{\grob}{Gr\"{o}bner}
	\newcommand{\Insert}{\mbox{\tt insert}}

\newsiamremark{remark}{Remark}
\newsiamremark{hypothesis}{Hypothesis}
\crefname{hypothesis}{Hypothesis}{Hypotheses}
\newsiamthm{claim}{Claim}

\headers{On $\bfmu$-Symmetric Polynomials}{J.~Yang and C.~Yap}

\title{On $\bfmu$-Symmetric Polynomials\thanks{
	Submitted to the editors on January 21, 2020.
\funding{
	Jing's work is supported by
	National Natural Science Foundation of China (Grant \#11801101)
	and Guangxi Science and Technology Program (Grant \#2017AD23056).
	Chee's work is supported by
	National Science Foundation (Grants CCF-1423228 and CCF-1564132).
	Chee is further supported under Chinese Academy of Science President's
	International Fellowship Initiative (2018),
	and Beihang International Visiting Professor Program No. Z2018060.
}}}

\author{Jing Yang\thanks{
	SMS-KLSE School of Mathematics and Physics,
	Guangxi University for Nationalities, Nanning, China.
  (\email{yangjing0930@gmail.com}).}
\and Chee K. Yap\thanks{
	Courant Institute of Mathematical Sciences, NYU, New York, USA.
  (\email{yap@cs.nyu.edu}).}
  }

\usepackage{amsopn}
\DeclareMathOperator{\diag}{diag}


%% file: p-fullpaper-musym.tex


	\ignore{%
    \begin{abstract}
	In this paper, we
	study functions of the roots of a univariate
	polynomial in which the roots have a given
	multiplicity structure $\bfmu$.
	Traditionally, root functions are studied via
	the theory of symmetric polynomials; we extend this
	theory to $\bfmu$-symmetric polynomials.
	We were motivated by a conjecture from Becker et al.~(ISSAC 2016)
	about the $\bfmu$-symmetry of a particular
	root function $\dplus(\bfmu)$, called D-plus.
	To investigate this conjecture,
    it was desirable to have fast algorithms for checking
    if a given root function is $\bfmu$-symmetric.
	We designed three such algorithms:
        one based on Gr\"{o}bner bases,
        another based on preprocessing and reduction,
        and the third based on solving linear equations.
    We implemented them in Maple and
	experiments show that the latter two
        algorithms are significantly faster than the first.

    \end{abstract}
	}%
\sect{Introduction}
	Suppose $P(x)\in \ZZ[x]$ is a polynomial with $m$ distinct
	complex roots $r_1\dd r_m$ where $r_i$ has multiplicity $\mu_i$.
	Write $\bfmu=(\mu_1\dd \mu_m)$ where
	we may assume $\mu_1\ge \mu_2\ge\cdots\ge\mu_m\ge1$.
	Thus $n=\sum_{i=1}^m \mu_i$ is the degree of $P(x)$.
	Consider the following function of the roots
		$$\dplus(P(x)) \as \prod_{1\le i<j\le m}
			(r_i-r_j)^{\mu_i+\mu_j}.$$
	Call this the \dt{D-plus} root function.
	The form of this root function\footnote{
		In \cite{becker+4:cluster:16}, the D-plus function was
		called a ``generalized discriminant'' and
		denoted by ``$\dstar(P(x))$'' or \dt{D-star}.
		On the suggestion of Prof.~Hoon Hong, we now reserve the D-star
		notation for the following root function
			$\dstar(P(x)) \as \prod_{1\le i<j\le m}
				(r_i-r_j)^{2\mu_i\mu_j}$.
		Unlike D-plus, it is easy to see that the D-star function
		is a rational function of the coefficients of $P(x)$.
	} was introduced by Becker et al \cite{becker+4:cluster:16}
	in their complexity analysis of a root clustering algorithm.
	The origin of this paper was to try
	to prove that $\dplus(P(x))$ is a rational
	function in the coefficients of $P(x)$.
	This result is needed for obtaining
	an explicit upper bound on the complexity
	of the algorithm on integer polynomials
	\cite{becker+3:cisolate:18}.
	This application is detailed in our companion
	paper \cite{yang-yap:dplus:20}.
	
	We may write ``$\dplus(\bfmu)$'' instead of $\dplus(P(x))$
	since the expression in terms of the roots $\bfr=(r_1\dd r_m)$
	depends only on the multiplicity structure $\bfmu$.
	For example, if $\bfmu=(2,1)$ then $\dplus(\bfmu)=(r_1-r_2)^3$
	and this turns out to be
		$$\left[a_1^3-(9/2)a_0a_1a_2+(27/2)a_0^2a_3\right]/a_0^3$$
		when $P(x)=\sum_{i=0}^3 a_{3-i}x^i$.
	More generally, for any function $F(\bfr)=F(r_1\dd r_m)$,
	we ask whether
	evaluating $F$ at the $m$ distinct roots of a polynomial $P(x)$ with
	multiplicity structure $\bfmu$ is rational in the coefficients of
	$P(x)$.
	\ignore{%
	In case $P(x)$ has only simple roots, the Fundamental Theorem
	of Symmetric Functions tells us the complete answer:
	$F(\bfr)$ is rational iff $F(\bfr)$ is a symmetric polynomial.
	 (This is wrong!  As we discovered (with Ruyong),
	 the ``converse'' is subtle.  Also, the premise of
	 "simply roots" is not really in the statement of the
	 Fundamental theorem).
	}%
	The Fundamental Theorem
	of Symmetric Functions gives a partial answer: {\em if
	$F(\bfr)$ is a symmetric polynomial then $F(\bfr)$
	is a rational function in the coefficients of $P(x)$.}
	This result does not exploit knowledge of the multiplicity
	structure $\bfmu$ of $P(x)$.
	We want a natural definition of ``$\bfmu$-symmetry''
    such that the following property is true:
	{\em if $F(\bfr)$ is $\bfmu$-symmetric, then
	$F(\bfr)$ is a rational function in the coefficients of $P(x)$.}
	When $\bfmu=(1\dd 1)$, i.e., all the roots of $P(x)$ are simple,
	then a $\bfmu$-symmetric polynomial is just
    a symmetric polynomial in the usual sense.
	So our original goal amounts to proving
	that $\dplus(\bfmu)$ is $\bfmu$-symmetric.
	It is non-trivial to check if any given
	root function $F$ (in particular
	$F=\dplus(\bfmu)$) is $\bfmu$-symmetric.
	We will designed three algorithms for this task.
	Although we feel that $\bfmu$-symmetry is a natural concept, to
	our knowledge, this has not been systematically studied before.

    The rest of this paper is organized as follows.
	In Section \ref{sec:def-mysym},
    we defined $\bfmu$-symmetric polynomials
    in terms of elementary symmetric polynomials
    and show some preliminary properties of such polynomials.
    In Section \ref{sec:specialcases},
    we proved the $\bfmu$-symmetry of $\dplus$ for some special
    $\bfmu$.
    To investigate the $\bfmu$-symmetry of $\dplus$ in the general case,
    three algorithms for checking $\bfmu$-symmetry
    are given in Sections \ref{sec:GBmethod}-\ref{sec:gist}.
    \revised{In Section \ref{sec:OtherBases},
    we discuss how to generalize the concepts and algorithms
    to other generators of symmetric polynomials
    different from the elementary symmetric polynomials.}
    In Section \ref{sec:experiments},
    we show experimental results from our Maple implementation
    of the three algorithms.%
    \revised{
    All the Maple code can be downloaded
    from \url{https://github.com/JYangMATH/mu-symmetry}.}
    We conclude in Section \ref{sec:conclusion}.

    \revised{
    The $\dplus$ conjecture is proved in a companion paper \cite{yang-yap:dplus:20}
    and an application is shown by giving an explicit complexity bound
    for root clustering.
    }

\sect{$\bfmu$-Symmetric Polynomials}\label{sec:def-mysym}
	Throughout the paper, assume $K$ is a field of characteristic $0$.
	For our purposes, $K=\QQ$ will do.
	We also fix three sequences of variables
		$$\bfx=(x_1\dd x_n),\quad
		\bfz=(z_1\dd z_n),\quad \bfr=(r_1\dd r_m)$$
	where $n\ge m\ge 1$.
    Intuitively, the $x_i$'s are roots (not necessarily distinct),
    $z_i$'s are variables
    representing the elementary symmetric functions
    of the roots, and $r_i$'s are the distinct roots.

	Let $\bfmu$ be a partition of $n$ with $m$ parts.  In other words,
	$\bfmu=(\mu_1\dd \mu_m)$ where $n=\mu_1+\cdots+\mu_m$
	and $\mu_1\ge \mu_2\ge\cdots\ge\mu_m\ge 1$.
	We denote this relation by
		$$\bfmu \vdash n.$$
	We call $\bfmu$ an \dt{$m$-partition} if it has $m$ parts.
	A \dt{specialization} $\sigma$ is any 	
	function of the form
		$\sigma:\set{x_1\dd x_n}\to\set{r_1\dd r_m}$.
	We say $\sigma$ is of \dt{type $\bfmu$} if
	$\#\sigma\inv(r_i)=\mu_i$ for $i=1\dd m$.
    Throughout the paper, we use
    $\#$ to denote the number of elements in a set, and
    $|\cdot|$ to denote the length of a sequence.
    In particular, $|\bfmu|=|\bfr|=m$.
	We say $\sigma$ is
	\dt{canonical} if $\sigma(x_i)=r_j$ and $\sigma(x_{i+1})=r_k$
	implies $j\le k$.  Clearly the canonical
	specialization of type $\bfmu$ is unique, and we may
	denote it by $\sigma_\bfmu$.
	
	Consider the polynomial rings $K[\bfx]$ and $K[\bfr]$.
	Any specialization $\sigma:\set{x_1\dd x_r}$ $\to\set{r_1\dd r_m}$
	can be extended naturally into a $K$-homomorphism
		$$\sigma: K[\bfx]\to K[\bfr]$$
	where $P=P(\bfx)\in K[\bfx]$ is mapped to
	$\sigma(P)= P(\sigma(x_1)\dd \sigma(x_n))$.
	When $\sigma$ is understood, we may write ``$\ol{P}$'' for
	the homomorphic image $\sigma(P)$.

	We denote the \dt{$i$-th elementary symmetric functions}
	($i=1\dd n$) in $K[\bfx]$ by $\e_i=\e_i(\bfx)$.
	For instance,
		\beqarrays
		\e_1 &\as & \sum_{i=1}^n x_i,\\
		\e_2 &\as & \sum_{1\le i<j\le n} x_ix_j,\\
			& \vdots & \\
		\e_n &\as& \prod_{i=1}^n x_i.
		\eeqarrays
	Also define
	$\e_0\as 1$.
    	Typically, we write $\ole_i$ for the $\sigma_\bfmu$ specialization
	of $e_i$ when $\bfmu$ is understood from the context;
	thus $\ole_i=\sigma_{\bfmu}(\e_i)\in K[\bfr]$.
	For instance, if $\bfmu=(2,1)$ then $\ole_1=2r_1+r_2$
	and $\ole_2=r_1^2+2r_1r_2$.
	%

        The key definition is the following:
        a polynomial $F\in K[\bfr]$ is said to be \dt{$\bfmu$-symmetric}
        if there is a symmetric polynomial $\whF\in K[\bfx]$ such that
        $\sigma_\bfmu(\whF)=F$.
        We call $\whF$ the \dt{$\bfmu$-lift} (or simply ``lift'') of
        $F$.  If $\ooF\in K[\bfz]$ satisfies $\ooF(\e_1\dd
		\e_n)=\whF(\bfx)$ then we call
		$\ooF$ the \dt{$\bfmu$-gist} of $F$.

    \begin{remark}
    (i)  We may also write $\lift{F}$ for any lift of $F$.
    Note that the $\bfmu$-lift and $\bfmu$-gist of $F$ are defined
    if and only if $F$ is $\bfmu$-symmetric.
    \\
    (ii) We view the $z_i$'s as symbolic representation of
    the symmetric polynomials $\e_i(\bfx)$'s.
    Moreover, we can write
		$\sigma_\bfmu(\ooF(\e_1\dd\e_n))$
		as
		$\ooF(\ole_1\dd\ole_n)$.
    \\(iii) Since $\ooF(\e_1\dd\e_n)$ is symmetric in $x_1\dd x_n$,
    we could use any specialization
    $\sigma$ of type $\bfmu$ instead of the canonical specialization
    $\sigma_\bfmu$,
    since
		$\sigma(\ooF(\e_1\dd\e_n))=
		 \sigma_\bfmu(\ooF(\e_1\dd\e_n))$.
    \\(iv) Although $\whF$ and $\ooF$ are mathematically equivalent,
    the gist concept lends itself to direct evaluation based on
    coefficients of $P(x)$.
    \end{remark}

    \begin{xample}\label{eg:def-musymfun}
    Let $\bfmu=(2,1)$ and $F(\bfr)=3r_1^2+r_2^2+2r_1r_2$.
    We see that $F(\bfr)$ is $\bfmu$-symmetric since
	$F(\bfr)=(2r_1+r_2)^2 -(r_1^2+r_1r_2)
	=\ole_1^2-\ole_2
	=\sigma_{\bfmu}(\e_1^2-\e_2)$.
    Hence lift of $F$ is
	$\whF=\e_1^2-\e_2=(x_1+x_2+x_3)^2-(x_1x_2+x_1x_3+x_2x_3)$
	and its gist is $\ooF(\bfz)=z_1^2-z_2$.
    \end{xample}

	We have this consequence of the Fundamental Theorem
	on Symmetric Functions:
		
		\bprol{mu-symmetry}
		Assume
			$$P(x)=\sum_{i=0}^nc_ix^{n-i}\in K[x]$$
		has $m$ distinct roots $\bfrho=(\rho_1\dd \rho_m)$
		of multiplicity $\bfmu=(\mu_1\dd \mu_m)$.
		\benum[(i)]
		\item If $F\in K[\bfr]$ is $\bfmu$-symmetric, then
			$F(\bfrho)$ is an element in $K$.
		\item If $\ooF\in K[\bfz]$ is the $\bfmu$-gist of $F$,
			then
			$$F(\rho_1\dd \rho_m)
	 			= \ooF\left(-c_1/c_0\dd(-1)^nc_n/c_0\right).$$
		\eenum
		\eprol
	\bpf
		{\small \beqarrys
		F(\bfr) &=& \sigma_\bfmu(\whF(\bfx))
						& \mbox{(by definition of $\bfmu$-symmetry)}\\
			&=& \sigma_\bfmu(\ooF(e_1\dd e_n))
					& \mbox{(by the Fundamental Theorem of Symmetric
							}\\
			&&&		\mbox{~Functions, as $\whF$ is symmetric)}\\
			&=& \ooF(\ole_1\dd \ole_n)
						& \mbox{(since $\ole_i =\sigma_\bfmu(e_i)$)}\\
		F(\bfrho) &=& \ooF(\ole_1(\bfrho)\dd \ole_n(\bfrho))\\
			&=& \ooF(-c_1/c_0 \dd (-1)^n c_n/c_0)
						& \mbox{(by Vieta's formula for roots)}
		\eeqarrys}
		This proves the formula in (ii).
		The assertion of (i) follows from the fact
		that $\ooF\in K[\bfz]$ and $c_i$'s belong to $K$.
	\epf

	\begin{xample}
		Consider the polynomial $F(r_1,r_2)$ in
		Example \ref{eg:def-musymfun}.  Suppose the polynomial
        $P(x)=c_0x^3+\cdots+c_3\in K[x]$ has two distinct roots
		$\rho_1$ and $\rho_2$ of multiplicities $2$ and $1$,
		respectively.  Then Proposition \ref{pro:mu-symmetry} says that
        $F(\rho_1,\rho_2)= 3\rho_1^2+\rho_2^2+2\rho_1\rho_2$
		is equal to
		$$\ooF(-c_1/c_0,c_2/c_0, -c_3/c_0)=
			\left(-c_1/c_0\right)^2-c_2/c_0\in K$$
		since $\ooF(z_1,z_2,z_3)=z_1^2-z_2$.
        \end{xample}

    \oldNew{
	In general, the converse of Proposition \ref{pro:mu-symmetry} is not
	true.  Assume $F\in K[\bfr]$ and $P(x)\in K[x]$ as in the
	Proposition.  Then
    with $m$ distinct root $\rho_1,\rho_2\dd\rho_m$
    of multiplicity $\bfmu$, $F(\rho_1,\rho_2\dd\rho_m)\in K$,
    $F$ may not be $\bfmu$-symmetric.
    For example, consider the case when $\mu_1>\cdots>\mu_m$.
    If $P(x)\in K[x]$, then its roots $\rho_i\in K~(1\leq i\leq m)$.
	Therefore, every polynomial $F(\bfr)\in K[\bfr]$ satisfies
	$F(\rho_1\dd\rho_m)\in K$.
	Apparently, not every polynomial in $K[\bfr]$ can be expressed as a
	rational
    function in coefficients of $P$. Thus there are some polynomials
    whose evaluation at $r_i=\rho_i~(1\leq i\leq m)$ are in $K$ and
    which is not $\bfmu$-symmetric.
	However, if $\bfmu=(1\dd 1)$, the converse of Proposition
	\ref{pro:mu-symmetry}
    is proved to be true.
    Both the above counter example and proof are given by Ruyong Feng.
    }{
	It is an interesting question to prove some converse
	of Proposition 1.  We plan to take this up
	in a future work.
    }

\ssect{On Lifts and the $\bfmu$-Ideal}\label{ssec:lift+ideal}
	We want to study the lift $\whF\in K[\bfx]$
	of a $\bfmu$-symmetric polynomial $F\in K[\bfr]$
	of total degree $\delta$.
	If we write $F$ as the sum of its homogeneous parts,
	$F=F_1+\cdots+ F_\delta$, then $\whF=\whF_1+\cdots+\whF_\delta$.
	Hence, we may restrict $F$ to be homogeneous.

	Next consider a polynomial $H(\bfz)\in K[\bfz]$.
	Suppose there is a \dt{weight function}
		$$\omega: \set{z_1\dd z_n} \to \NN=\set{1,2,\ldots}$$
	then for any term $t=\prod_{i=1}^n z_i^{d_i}$,
	its \dt{$\omega$-degree} is $\sum_{i=1}^n d_i \omega(z_i)$.
	Normally, $\omega(z_i)=1$ for all $i$; but in
	this paper, we are also interested in the weight
	function where $\omega(z_i)=i$.  For short, we simply
	call this $\omega$-degree of $t$ its \dt{weighted degree},
    \revised{denoted by $\wdeg(t)$}.
	The weighted degree of a polynomial $H(\bfz)$
	is just the maximum weighted degree of terms in its support,
    \revised{denoted by $\wdeg(H)$}.
	A polynomial $H(\bfz)$ is said to be \dt{weighted homogeneous}
    \revised{or \dt{$\omega$-homogeneous}}
	if all of its terms have the same weighted degree.
    Note that the weighted degree of
    a polynomial $H\in K[\bfz]$ is the same as the degree of
    $H(\e_1\dd\e_n)\in K[\bfx]$.

	The gist $\ooF$ of $F$ is not unique:
	for any gist $\ooF$, we can decompose it as $\ooF=\ooF_0+\ooF_1$
	where $\ooF_0$ is the weighted homogeneous part of $\ooF$
	of degree $\delta$,
	and $\ooF_1\as \ooF-\ooF_0$.  Then $\ooF(\ole_1\dd\ole_n)=F$
	implies that $\ooF_0(\ole_1\dd\ole_n)=F$
	and $\ooF_1(\ole_1\dd\ole_n)=0$.
	We can always omit $\ooF_1$ from the gist of $F$.
	We shall call any polynomial $H(\bfz)\in K[\bfz]$
	a \dt{$\bfmu$-constraint} if
	$H(\ole_1\dd\ole_n)=0$.
    Thus, $\ooF_1$ is a
	$\bfmu$-constraint.

	It follows that when trying to check if $F$ is
	$\bfmu$-symmetric, it is sufficient to look for gists $\ooF$
	among weighted homogeneous polynomials of the same degree as $F$,
	i.e., $\delta$.
	But even this restriction does not guarantee uniqueness
	of the gist of $F$
	because there could be $\bfmu$-constraints of weighted homogeneous
	degree $\deg(F)$.
	To illustrate this phenomenon, we consider the following example.

    \begin{xample}
	Let $\bfmu=(2,2)$.
    	Consider the polynomial $F=r_1^3+2r_1^2r_2+2r_1r_2^2+r_2^3$.
    	It is easy to verify that both
    	$\whF=\frac{1}{8}\e_1^3-\frac{1}{2}\e_3$ and
    	$\whF'=\frac{1}{2}\e_1\e_2-\frac{3}{2}\e_3$
	are the lifts of $F$. Therefore,
        $\ooF=\frac{1}{8}z_1^3-\frac{1}{2}z_3$ and
    	$\ooF'=\frac{1}{2}z_1z_2-\frac{3}{2}z_3$
    are the gists of $F$.
	It follows that the difference
    $$H=\ooF-\ooF'=\efrac{8}\Big(z_1^3+8z_3-4z_1z_2)$$
	is a $\bfmu$-constraint.
    We may check that
	$$
	H(\ole_1\dd\ole_4)
        	= \efrac{8}(2r_1+2r_2)^3
		+(2r_1^2r_2+2r_1r_2^2)
		-\efrac{2}(2r_1+2r_2)(r_1^2+4r_1r_2+r_2^2)
		=0.
	$$
    \end{xample}

    It is easy to check that the set of all $\bfmu$-constraints
    forms an ideal in $K[\bfz]$
    which we may call the \emph{$\bfmu$-ideal},
    denoted by $\calJ_{\bfmu}$.
    \revised{
    Note that $H(\ole_1\dd\ole_n)$ is in $K[\bfr]$
    but $H$ is in $K[\bfz]$.
    So we introduce an ideal in $K[\bfz,\bfr]$
    to connect them:
    \beql{cali}
         \calI_\bfmu \as \bang{z_1-\ole_1\dd z_n-\ole_n}.
    \eeql
    Actually $\calJ_{\bfmu}$ can be generated by $\calI_\bfmu$
    as indicated by Theorem \ref{thm:D-ideal}.}

	\begin{xample}
    The following set of polynomials generates the $(2,2)$-ideal:
    	\beqarrays
		G_{3}: && z_1^3-4z_1z_2+8z_3\\
		G_{4}: && z_1^2z_2+2z_1z_3-4z_2^2+16z_4\\
		G_{5}: && z_1^2z_3+8z_1z_4-4z_2z_3\\
		G_{6}: && z_1^2z_4-z_3^2\\
		G_{7}: && 4z_1z_2z_4-z_1z_3^2-8z_3z_4\\
		G_{8}:&& 2z_1z_3z_4-4z_2^2z_4+z_2z_3^2+16z_4^2\\
		G_9:&& 8z_1z_4^2-4z_2z_3z_4+z_3^3\\
		G_{10}:&& z_1z_3^3-8z_2^3z_4+2z_2^2z_3^2+32z_2z_4^2+8z_3^2z_4\\
		G_{12}:&& 16z_2^2z_4^2-8z_2z_3^2z_4+z_3^4-64z_4^3.
	   \eeqarrays
	    We computed this by first computing the
	    Gr\"{o}bner basis of the ideal
	    \begin{align*}
	    \left<
	    z_1-\ole_1,
	    z_2-\ole_2,
	    z_3-\ole_3,
	    z_4-\ole_4
	    \right>
	    &=\\
	    \left<z_1-(2r_1+2r_2),z_2-(r_1^2\right.&
	    \left.+4r_1r_2+r_2^2),
	    z_3-(2r_1^2r_2+2r_1r_2^2),
	    z_4-r_1^2r_2^2
	    \right>.
	    \end{align*}
	By Theorem \ref{thm:D-ideal},
    the restriction of the Gr\"obner basis to $K[\bfz]$ is
	the above set of generators.
    \end{xample}

\ssect{Examples of $\bfmu$-symmetric Polynomials}
    Although $\bfmu$-symmetric polynomials originated
    from symmetric polynomials, they differ in many ways
    as seen in these examples.

	\bitem
	\item
		A $\bfmu$-symmetric polynomial need not be symmetric.
		Let $\bfmu=(2,1)$ and $n=2+1=3$.  Then
		$2r_1+r_2$ is $\bfmu$-symmetric whose lift is $e_1$,
		but it is not symmetric.

	\item
		A symmetric polynomial need not be $\bfmu$-symmetric.
		Consider the symmetric polynomial $F=r_1+r_2\in K[r_1,r_2]$.
		It is not $\bfmu$-symmetric with $\bfmu=(2,1)$.
		If it were, then there is a linear symmetric
		polynomial $\whF=ce_1$ such that $\sigma_{\bfmu}(\whF)=r_1+r_2$.
		But clearly such $\whF$ does not exist.

        \item  Symmetric polynomials can be $\bfmu$-symmetric.
	    Note that $(r_1-r_2)^2$ is obviously symmetric
	    in $K[r_1,r_2]$.
	    According to \refLem{sdisc},
	    it is also $\bfmu$-symmetric for any
	    $\bfmu=(\mu_1,\mu_2)$.
    \eitem

	In the following we will use this notation:
	$[n] \as \set{1\dd n}$, and
	let $[n] \choose k$ denote the set of all $k$-subsets of $[n]$.
	For $k=0\dd n-2$, we may define the function
		\beql{snk}
			S^n_k=S^n_k(\bfx) \as \sum_{I\in {[n]\choose n-k}}
				\prod_{i\neq j\in I} \Big(x_i-x_j\Big)^2
		\eeql
	called the \dt{$k$th subdiscriminant}
	in $n$ variables.  By extension, we could also define $S^n_{n-1}=1$.
	When $k=0$, we have
		$S^n_0= \prod_{i\neq j\in [n]} \Big(x_i-x_j\Big)^2$.
	In applications, the $x_i$'s are roots
	of a polynomial $P(x)$ of degree $n$, and $S^n_0$ is
	the standard discriminant of $P(x)$.
    Clearly $S^n_{k}$ is a symmetric polynomial in $\bfx$.
	\ignore{%
    Therefore there is a polynomial $T^n_k\in K[\bfz]$
	such that
	\beql{tnk}
		T^n_k(\e_1\dd \e_n)=S^n_k.
	\eeql
    }

	\bleml{sdisc}
		Define $\Delta \as \prod_{1\le i<j\le m}(r_i-r_j)^2$.
		\\ (a) $\Delta$
		is $\bfmu$-symmetric with lift given by
			$$\wh{\Delta} = \efrac{\prod_{i=1}^m\mu_i}
				\cdot S^n_{n-m}$$
		where
		$S^n_{n-m}\in K[\bfx]$
		is the $(n-m)$-th subdiscriminant.
		\\ (b)
		In particular, when $m=2$, we have an explicit formula
		for the lift of $\Delta$:
		    $$\wh{\Delta}=\frac{(n-1)e_1^2-2ne_2}{\mu_1\mu_2},$$
        	where $n=\mu_1+\mu_2$.
    \eleml
    \bpf
        Let $\bfmu=(\mu_1\dd \mu_m)$.
        Consider the $m$-th subdiscriminant $S^n_m$ in $n$ variables.
        We may verify that
        	$$\sigma_{\bfmu}(S^n_{n-m})=\Delta \cdot \prod_{i=1}^m\mu_i.$$
        This is equivalent to
		$$\sigma_{\bfmu}\left(\frac{1}{\prod_{i=1}^m\mu_i}\cdot
		S^n_{n-m}\right) =\Delta.$$
        Therefore, $\efrac{\prod_{i=1}^m\mu_i}S^n_{n-m}$ is
        the $\bfmu$-lift of $\Delta$.

        To obtain the explicit formula in the case $m=2$,
        consider the symmetric polynomial \\
        $Q\as\sum_{i<j}(x_i-x_j)^2$.
        It is easy to check that $Q=(n-1)\e_1^2 - 2n\e_2$.
        A simple calculation shows that
        	\[\sigma_{\mu}(Q)=\mu_1\mu_2(r_1-r_2)^2.\]
        Thus, we may choose
        $\wh{\Delta}=\frac{(n-1)e_1^2 - 2n e_2}{\mu_1\mu_2}$.
    \epf

\sect{Explicit Formulas for Special Cases of $\dplus$}\label{sec:specialcases}
    The following two theorems show the $\bfmu$-symmetry of
    some special $\dplus$ polynomials.
    In other words, they confirmed our conjecture about $\dplus$.

	\bthm
		There exists $\ooF_n\in K[\bfz]$ such that
		for all $\bfmu$
		satisfying $\bfmu=(\mu_1,\mu_2)$ and $\mu_1+\mu_2=n$,
		we have
			$$\ooF_n(\ole_1,\ole_2)=\dplus(\bfmu).$$
		More explicitly,
		\bitem
		\item $n$ is even:
			$\ooF_n = \Big( \frac{(n-1)z_1^2-2n z_2}{
				\mu_1\mu_2}\Big)^{n/2}$
		\item $n$ is odd:
		\beqarrays
			\ooF_n &=&
			\Big( \frac{(n-1) z_1^2-2n z_2}{
				\mu_1\mu_2}\Big)^{\frac{n-3}{2}}
			\Big( k_1 z_1^3
				+ k_2 z_1 z_2+k_3 z_3\Big)
		\eeqarrays
		where $k_1= \frac{-(n-1)(n-2)}{d},
			k_2= \frac{3n(n-2)}{d},
			k_3= \frac{-3n^2}{d}$
			and $d=\mu_1\mu_2(\mu_1-\mu_2)$.
		\eitem
	\ethm	
	\bpf
    From \refLem{sdisc}(b),
    we know that $(r_1-r_2)^2$ is $\bfmu$-symmetric for arbitrary
    $n$ and
    $$(r_1-r_2)^2=\frac{(n-1)\ole_1^2-2n\ole_2}{\mu_1\mu_2}.$$

    When $n$ is even,
    \begin{align*}
    \dplus(\bfmu)=\left((r_1-r_2)^2\right)^{\frac{n}{2}}
        &=\left(\frac{(n-1)\ole_1^2-2n\ole_2}
            {\mu_1\mu_2}\right)^{\frac{n}{2}}\\
        &=\left(\frac{(n-1)\ole_1^2-2n\ole_2}
	    {\mu_1\mu_2}\right)^{\frac{n}{2}}=\ooF_n(\ole_1,\ole_2).
    \end{align*}
    Thus the case for even $n$ is proved.
    It remains to prove the case for odd $n$.
    First, it may be verified that
       \[(r_1-r_2)^3 =k_1\ole_1^3+k_2\ole_1\ole_2+k_3\ole_3,\]
    where
    \[k_1= \frac{-(n-1)(n-2)}{d},\quad
		k_2= \frac{3n(n-2)}{d},\quad
		k_3= \frac{-3n^2}{d}~~\mbox{and}~~
		d=\mu_1\mu_2(\mu_1-\mu_2).\]
    \ignore{
	Substitution of
    \begin{align*}
    \ole_1^3&=(\mu_1r_1+\mu_2r_2)^3\\
    &=\mu_1^3r_1^3+3\mu_1^2\mu_2r_1^2r_2
                    +3\mu_1\mu_2^2r_1r_2^2+\mu_2^3r_2^3,\\
    \ole_1\ole_2&=(\mu_1r_1+\mu_2r_2)
        \left({\mu_1 \choose 2}r_1^2+\mu_1\mu_2r_1r_2
                    +{\mu_2 \choose 2}r_2^2\right)\\
    &=\frac{1}{2}\left[(\mu_1^3-\mu_1^2)r_1^3+(3\mu_1^2\mu_2-\mu_1\mu_2)r_2r_1^2
        +(3\mu_1\mu_2^2-\mu_1\mu_2)r_2^2r_1+(\mu_2^3-\mu_2^2)r_2^3\right]\\
    \ole_3&={\mu_1 \choose 3}r_1^3+{\mu_1 \choose 2}\mu_2r_1^2r_2
                    +\mu_1{\mu_2 \choose 2}r_1r_2^2+{\mu_1 \choose 3}r_1^3
    \end{align*}
    into $k_1\ole_1^3+k_2\ole_1\ole_2+k_3\ole_3$
    will lead to
    \[k_1\ole_1^3+k_2\ole_1\ole_2+k_3\ole_3
        =c_0r_1^3+c_1r_1^2r_2+c_2r_1r_2^2+c_3r_2^3
    \]
    where
    \begin{align*}
    c_0&=\left(k_1+\frac{1}{2}k_2+\frac{1}{6}k_3\right)\mu_1^3
            -\frac{1}{2}\left(k_2+k_3\right)\mu_1^2
            +\frac{1}{3}k_3\mu_1,\\
    c_1&=\left(3k_1+\frac{3}{2}k_2+\frac{1}{2}k_3\right)\mu_1^2\mu_2
            -\left(\frac{1}{2}k_2+\frac{1}{2}k_3\right)\mu_1\mu_2,\\
    c_2&=\left(3k_1+\frac{3}{2}k_2+\frac{1}{2}k_3\right)\mu_1\mu_2^2
            -\left(\frac{1}{2}k_2+\frac{1}{2}k_3\right)\mu_1\mu_2,\\
    c_3&=\left(k_1+\frac{1}{2}k_2+\frac{1}{6}k_3\right)\mu_2^3
            -\frac{1}{2}\left(k_2+k_3\right)\mu_2^2
            +\frac{1}{3}k_3\mu_2,\\.
    \end{align*}
    Then evaluate $c_0,c_1,c_2,c_3$ at
    \[k_1= \frac{-(n-1)(n-2)}{d},\quad
		k_2= \frac{3n(n-2)}{d},\quad
		k_3= \frac{-3n^2}{d}~~\mbox{and}~~
		d=\mu_1\mu_2(\mu_1-\mu_2)\]
    and we get
    $c_0=1,c_1=-3,c_2=3,c_3=1$,
    i.e.,
    $k_1\ole_1^3+k_2\ole_1\ole_2+k_3\ole_3=(r_1-r_2)^3$.
    }

    It follows that
    \begin{align*}
    \dplus(\bfmu)&=\left((r_1-r_2)^2\right)^{\frac{n-3}{2}}(r_1-r_2)^3\\
        &=\left(\frac{(n-1)\ole_1^2-2n\ole_2}
            {\mu_1\mu_2}\right)^{\frac{n-3}{2}}
          \left(k_1\ole_1^3 + k_2\ole_1\ole_2+k_3\ole_3\right)\\
        &=\left(\frac{(n-1)\ole_1^2-2n\ole_2}
            {\mu_1\mu_2}^{\frac{n}{2}}\right)
	    \left(k_1\ole_1^3 +
	    k_2\ole_1\ole_2+k_3\ole_3\right)\\
	&=\ooF_n(\ole_1,\ole_2,\ole_3)
    \end{align*}
    where
    \[k_1= \frac{-(n-1)(n-2)}{d},\quad
		k_2= \frac{3n(n-2)}{d},\quad
		k_3= \frac{-3n^2}{d}~~\mbox{and}~~
		d=\mu_1\mu_2(\mu_1-\mu_2).\]
	\epf

    Another special case of $\dplus(\bfmu)$
    is where $\bfmu=(\mu,\mu\dd \mu)$.

    \bthm
    If all $\mu_i$'s are equal to $\mu$, then
    $\dplus(\bfmu)$ is $\bfmu$-symmetric with lift given by
	$\whF_n(\bfx)
		=\left(\frac{1}{\mu^{m}}\cdot S^n_{n-m}\right)^\mu$
    where $S^n_{n-m}$ is given by \refLem{sdisc}(a).
    \ethm

    \bpf
    Since $\mu_i=\mu~(1\leq i\leq m)$,
	    $$\dplus(\bfmu)
	    	=\prod_{i<j}(r_i-r_j)^{2\mu}
		=\left(\prod_{i<j}(r_i-r_j)^2\right)^\mu.$$
    This expression for $\dplus$ is $\bfmu$-symmetric
    since $\prod_{i<j}(r_i-r_j)^2$ is $\bfmu$-symmetric
    by \refLem{sdisc}(a).
    Moreover,
    \refLem{sdisc}(a) also shows that the lift of
    $\prod_{i<j}(r_i-r_j)^2$ is
    $\frac{1}{\mu^m}\cdot S^n_{n-m}$.
    Thus we may choose
    $\whF_n=\left(\frac{1}{\mu^m}\cdot S^n_{n-m}\right)^\mu$.
    \epf

    The following example shows two ways to compute $\dplus$.
    One is using the definition and the other is using
    the formula of $\dplus$ in coefficients.

        \begin{xample}
	    Let $P(x)\!=\!(x^2-x-1)^2(x-1)\!=\!
	    		(1,\!-3,\!1,\!3,\!-1,\!-1)\cdot (x^5,x^4\dd x,1)^T$.
		Then $(\rho_1,\rho_2,\rho_3)=(\phi,\wh{\phi},1)$ are the roots
		with multiplicity $\bfmu=(2,2,1)$.  Here $\phi=(1+\sqrt{5})/2$
		is the golden ratio and $\wh{\phi}=1-\phi$ is its conjugate.
		It turns out that in this case, $\dplus(\bfmu) = -25$
		as directly computed from the formula in the roots
		$(\rho_1,\rho_2,\rho_3)$.
		We can also compute it using the gist
		$\mathring{D}^+(\bfz)$ of $\dplus$,
		i.e., $\dplus(\bfmu)=\mathring{D}^+(\ole_1\dd\ole_5)$.
		Here is the gist of $\dplus$
		(which can be obtained from our algorithms below):
        \begin{align*}
	    \mathring{D}^+(z_1,z_2,z_3,z_4,z_5)=
	    	&\frac{10125}{4}z_5^2-\frac{11}{2}z_1^2z_2z_3^2
              	-3z_1^4z_2z_4+67z_1^3z_3z_4-207z_1^3z_2z_5\\
             &+\frac{2517}{4}z_1z_2^2z_5+171z_1^2z_3z_5
             -\frac{5955}{4}z_2z_3z_5 +\frac{615}{2}z_1z_4z_5\\
             &-184z_2z_4^2+12z_1^5z_5+z_1^4z_3^2+6z_2^2z_3^2
              +\frac{9}{2}z_1z_3^3+48z_2^3z_4\\
             &+\frac{1737}{4}z_3^2z_4+\frac{277}{4}z_1^2z_4^2
             -\frac{1255}{4}z_1z_2z_3z_4.
        \end{align*}
		According to Vieta's formula for $n=5$,
		$(\ole_1\dd \ole_5)\!=\!(-c_1,c_2, -c_3,c_4,-c_5)\!=\!
            	(3,1,-3,-1,1)$.  Then,
	    substituting $z_i$ by $\ole_i = (-1)^i c_i$,
	    we also obtain $\dplus(2,2,1)=-25$.
        \end{xample}

\sect{Computing Gists via Gr\"{o}bner Bases}
\label{sec:GBmethod}
    In this section, we consider a \grob\ basis
    algorithm to compute the
    $\bfmu$-gist of a given polynomial $F\in K[\bfr]$,
    or detect that it is not $\bfmu$-symmetric.
	In fact, we first generalize our concept of gist:
	fix an arbitrary (ordered) set
			$$\calD=(d_1\dd d_\ell),\quad d_i\in K[\bfr].$$
	Call $\calD$ the basis.
	If $F\in K[\bfr]$ and
		$\ooF\in K[\bfy]$ where $\bfy=(y_1\dd y_\ell)$ are
	$\ell$ new variables,
	then $\ooF(\bfy)$ is called a \dt{$\calD$-gist} of $F$
	if $F(\bfr)=\ooF(d_1\dd d_\ell)$.
	Note that if $\calD=(\ole_1\dd \ole_n)$ (so $\ell=n$)
	then
	a $\calD$-gist is just a $\bfmu$-gist
	(after renaming $\bfy$ to $\bfz$).
	
	We now give a method to compute a $\calD$-gist of $F$
	using Gr\"obner bases.  To this end, define the ideal
		$$\calI_\calD \as \bang{v_1\dd v_\ell} \ib K[\bfr,\bfy]$$
	where $v_i \as y_i-d_i$.
    Moreover, let $\calG_\calD$ be the
    \grob\ basis of $\calI_\calD$ relative to the
	the term ordering $\prec_{ry}$.
	The ordering is defined as follows:
		$$\bfr^\bfalpha\bfy^\bfbeta \prec_{ry}
				\bfr^{\bfalpha'}\bfy^{\bfbeta'}$$
	iff
		$\bfr^\bfalpha \prec_r \bfr^{\bfalpha'}$
	or else $\bfalpha=\bfalpha'$ and
		$\bfy^\bfbeta \prec_y \bfy^{\bfbeta'}$.
	Here $\prec_r$ and $\prec_y$ are term orderings
	in $K[\bfr]$ and $K[\bfy]$ respectively.
	Note that
	$\prec_{ry}$ is called the lexicographic product
	of $\prec_r$ and $\prec_y$ in \cite[\S 12.6]{yap:algebra:bk}.
	We have two useful lemmas. The first is about the
	ideal $\calI_\calD$, and the second about its Gr\"obner basis
	$\calG_\calD$.
	\bleml{basic}
		For all $R\in K[\bfy]$,
		$$R(\bfy)-R(\calD)\in \calI_\calD.$$
	\eleml
	\bpf
		Consider any term $\bfy^\bfalpha$
		where $\bfalpha=(\alpha_1\dd\alpha_\ell)$.
		Its image in the quotient ring $K[\bfy]/\calI_\calD$ is:
			{\small\beqarrays
			\bfy^\bfalpha +\calI_\calD
				&=& \Big(\prod_{i=1}^\ell y_i^{\alpha_i}\Big)
						+\calI_\calD\\
				&=& \Big(\prod_{i=1}^\ell (d_i+(y_i-d_i))^{\alpha_i}\Big)
						+\calI_\calD\\
				&=& \Big(\prod_{i=1}^\ell d_i^{\alpha_i}+\calI_\calD\Big)
						+\calI_\calD\\
				&=& \Big(\prod_{i=1}^\ell d_i^{\alpha_i}\Big)
						+\calI_\calD\\
				&=& \calD^\bfalpha +\calI_\calD.
			\eeqarrays}
		Thus
			$\bfy^\bfalpha-\calD^\bfalpha\in \calI_\calD$.
		Since
			$R(\bfy)-R(\calD)$ is a linear combination
			of $\bfy^\bfalpha-\calD^\bfalpha$'s,
		our lemma is proved.
	\epf

	By a \dt{weighted homogeneous ideal} we mean one that
	is generated by weighted homogeneous polynomials.
	The following is a generalization of
	\cite[Theorem 12.20, p.385]{yap:algebra:bk},
	where the result is stated for homogeneous ideals.

	The following is a consequence of
	\cite[Theorem 12.21, p.387]{yap:algebra:bk}:
	\blem
		$\calG_\calD\cap K[\bfy]$ is a Gr\"obner basis
		for the elimination ideal
		$\calI_\calD\cap K[\bfy]$ with respect to
		the term ordering $\prec_y$.
	\elem

    \revised{
    If $R(\calD)=0$, then $R(\bfy)$ is called a $\calD$-constraint,
    which generalizes the concept of $\bfmu$-constraint.
    Similar to $\bfmu$-constraints,
    one may verify that
    all $\calD$-constraints forms an ideal, denoted by $\calJ_{\calD}$.
    Then we have the following theorem.

    \bthml{D-ideal}
    $\calJ_{\calD}=\calI_\calD \cap K[\bfy]$.
    \ethml

    \bpf We will prove the theorem with the following two inclusions.
    \bitem
    \item $\calJ_{\calD}\subseteq \calI_\calD \cap K[\bfy]$.

    Consider any $R\in\calJ_{\calD}$.
    Then $R(\calD)=0$ implies $R(\bfy)=R(\bfy)-R(\calD)\in \calI_\calD$
    by \refLem{basic}.

    \item $\calJ_{\calD}\supseteq \calI_\calD \cap K[\bfy]$.

    For any $R\in \calI_\calD \cap K[\bfy]$, $R\in \calI_\calD$.
    Thus there exist $B_1\dd B_n\in K[\bfr,\bfy]$ such that
    \[R(\bfy)=\sum_{i=1}^n(y_i-d_i)\cdot B_i.\]
    Substitution of $y_i=d_i$ leads to $R(d_1\dd d_n)=0$,
    which implies $R\in\calJ_{\calD}$.
    \eitem
    \epf
    }

    The following is a generalization of Proposition 4~in Cox
    \cite[Chapter 7, Section 1]{cox-et-al:ideals:bk}
	(except for claims about uniqueness):

    \bthml{normalform}
	Fix the above \grob\ basis $\calG_\calD$.
    Let $R\in K[\bfr,\bfy]$ be the normal form of $F\in K[\bfr]$
    	relative to $\calG_\calD$.
	\benum[(i)]
	\item If $R\in K[\bfy]$, then $R$ is a $\calD$-gist of $F$.
	\item If $F$ has a $\calD$-gist, then $R\in K[\bfy]$.
	\eenum
    \ethml
    \bpf
	In the following, we use the
	specialization $\sigma: y_i\mapsto d_i$
	for all $i$.  This induces the
	homomorphism $\sigma:K[\bfr,\bfy]\to K[\bfr]$ taking
	every polynomial $f(\bfr,\bfy)$ in the ideal $\calI_\calD$ to $0$,
	i.e., $\sigma(f)=0$.
	\benum[(i)]
	\item
		Since $R$ is the normal form of $F$,
		$F-R\in \calI_\calD$.
		Thus $\sigma(F-R)=0$ or $\sigma(F)=\sigma(R)$.
		But $F\in K[\bfr]$ implies $\sigma(F)=F$.
		The assumption that $R\in K[\bfy]$ implies that
		$\sigma(R) =R(\calD)=R(d_1\dd d_\ell)$.
		We conclude that $R$ is a $\calD$-gist of $F$:
			$$F(\bfr)  = R(\calD)$$
	\item
		By assumption, $F$ has a $\calD$-gist
		$\ooF\in K[\bfy]$, i.e., $\ooF(\calD)=F$.
		Let $\wtR$ be the normal form of $\ooF$.
		CLAIM: $R-\wtR\in \calI_\calD$.
		To see this, we write $R-\wtR$ as a sum
			$$R-\wtR = (R-F)+(F-\ooF) + (\ooF-\wtR).$$
		We only need to verify that each of the three summands
		belong to $\calI_\calD$:
		in part (i), we noted that $R-F\in \calI_\calD$;
		the third summand $\ooF-\wtR\in \calI_\calD$ for the
		same reason. The second summand $F-\ooF\in\calI_\calD$
		by an application of \refLem{basic}.
		To conclude that $R\in K[\bfy]$, we assume (by
		way of contradiction) that $R\notin K[\bfy]$.
		By our choice of term ordering for $\calG_\calD$,
		we know that $\lt(R-\wtR)=\lt(R)$.
		But $R-\wtR\in \calI_\calD$ implies that there is
		polynomial $g\in \calG_\calD$ such that $\lt(g)|\lt(R)$.
		This contradicts the fact that $R$ is a normal form.
	\eenum
    \epf

    \revised{
    Now we consider the special case
    when $y_i-d_i$ is weighted homogeneous relative to a weight function:
    \[\omega:(\bfy,\bfr)\rightarrow \mathbb{N}.\]
    A set of polynomials is said to be \dt{weighted homogeneous}
    or \dt{$\omega$-homogeneous}
    if every polynomial in the set is $\omega$-homogeneous.
    Let $K_{\omega}[\bfy,\bfr]$ denote the set of
    all the $\omega$-homogeneous polynomials in $K[\bfy,\bfr]$.
    It is obvious that
    all polynomials in $K_{\omega}[\bfy,\bfr]$
    of weighted degree $\delta$ form a $K$-vector space,
    denoted by $K_{\omega}^\delta[\bfy,\bfr]$
    where we assume $0\in K_{\omega}^\delta[\bfy,\bfr]$.
    Therefore, we define the weighted degree of $0$ to be $\delta$
    when $0$ is viewed as an element in $K_{\omega}^\delta[\bfy,\bfr]$.
    When $\omega\equiv 1$,
    we simplify $K_{\omega}^\delta[\bfy,\bfr]$ into $K^\delta[\bfy,\bfr]$.

    Assume $f,f'\in K_{\omega}[\bfy,\bfr]$.
    Then the following properties can be easily verified.
    \benum[(i)]
    \item $\wdeg(f,f')=\max(\wdeg(f), \wdeg(f'))$.

    \item $\wdeg(f\pm f')\le\max(\wdeg(f), \wdeg(f'))$.

    \item $\wdeg(f\cdot f')=\wdeg(f)+\wdeg(f')$.

    \item The S-polynomial of $f$ and $f'$ is weighted homogeneous.

    \item If $\calG\subseteq K[\bfy,\bfr]$ is a weighted homogeneous
    Gr\"obner basis,
    then the normal form of $f$ relative to $\calG$ is weighted homogeneous
    of weighted degree $\wdeg(f)$.

    \item If $\calF\subseteq K[\bfy,\bfr]$ is weighted homogeneous,
    so is the Gr\"obner basis of $\calF$.


    \item Let $F=\sum_{i=0}^\delta F_i\in K[\bfy,\bfr]$
    where $\wdeg(F_i)=i$
    and $\calG\ib K_{\omega}[\bfy,\bfr]$ is a Gr\"obner basis.
    Then the normal form of $F$ relative to $\calG$ is
    the sum of the normal form of $F_i$ relative to $\calG$.
    \eenum

    If $\calF\subseteq K[\bfy,\bfr]$ is weighted homogeneous,
    we say a polynomial $H\in K[\bfy,\bfr]$ is \dt{$\calF$-minimal}
    if for all $H'\in K[\bfy,\bfr]$,
    \[H\equiv H'~(\modd~\calI_{\calF})\quad\text{implies}\quad
    \wdeg(H)\le\wdeg(H').\]
    Then we have the following lemma.

    \bleml{minimal_degree}
    If $\calG\ib K_\omega[\bfy,\bfr]$ is a Gr\"obner basis
    and $F$ is weighted homogeneous,
    then the normal form of $F$ relative to $\calG$
    is $\calG$-minimal, i.e.,
    for any $F' \equiv F~(\modd~I_\calG)$,
    $\wdeg(F')\ge\wdeg(F)$.
    \eleml

    \bpf
    Note that any $F'\in K[\bfy,\bfr]$ can be decomposed into
    weighted homogeneous components, i.e.,
    $F'=\sum_{i}F_i'$
    where $F_i'$ is weighted homogeneous.
    Let $R'$ and $R_i'$ be the normal forms of $F'$ and $F_i'$
    relative to $\calG$ respectively. Then
    $R'=\sum_{i}R_i'$.
    Let $R$ be the normal form of $F$ relative to $\calG$.
    Then there exists $i$ such that $R_i'=R$ and $R_j'=0$ if $j\ne i$.
    Therefore,
    \[\wdeg(F')\ge\wdeg(F_i')=\wdeg(R_i')=\wdeg(R)=\wdeg(F).\]
    The lemma is proved.
    \epf

    \bthml{minimalgist}
    If $R\in K[\bfy]$ is the normal form of $F\in K[\bfr]$
    relative to $\calG_\calD$ where $\calD$ is weighted homogeneous
    and $\calG_\calD$ is the Gr\"obner basis of
    \[\calI_\calD=\bang{y_1-d_1\dd y_\ell-d_\ell},\]
    then $R$ is a minimal $\calD$-gist of $F$.
    \ethml

    \bpf
    First by Theorem \ref{thm:normalform}, $R$ is a $\calD$-gist of $F$.
    If $\calD$ is weighted homogeneous, so is $\calG_\calD$.
    By Lemma \ref{lem:minimal_degree}, $R$ is $\calG$-minimal.
    Since $I_{\calG}=I_\calD$, $R$ is $\calD$-minimal by definition.
    \epf

    Theorems \ref{thm:normalform} and \ref{thm:minimalgist}
    lead to the following algorithm after specializing
    $y_i-d_i$ to $z_i-\ole_i$ and $\omega$ to
    \[
    \omega(z_i)=i,\quad\omega(r_i)=1.
    \]
    }

    \begin{figure}[h]
    \Ldent
    \progb{\\
	    $\ggist(F,\bfmu)$:
	    \lline[0] \INPUt:~~ $F\in K^\delta[\bfr]$ and
	    	$\bfmu=(\mu_1\dd\mu_m)$.
	    \lline[0] \OUTPUt: a minimal $\bfmu$-gist of $F$
	    		or say ``$\ooF$ does not exist"
	    \lline[8] $\calB\ass \left\{z_1-\ole_1(\bfr)\dd
                    z_n-\ole_n(\bfr)\right\}$
	    \lline[8] $ord\ass plex(r_1\dd r_m,z_1\dd z_n)$
	    \lline[8] $\calG\ass GroebnerBasis(\calB, ord)$
	    \lline[8] $R\ass NormalForm(F,\calG,ord)$
	    \lline[8] If $\deg(R, \bfr)>0$ then
	    \lline[16] Return ``$\ooF$ does not exist''
	    \lline[8] Else
	    \lline[16] Return $R$
		\Ldent
	}
    \caption{The \ggist\ algorithm.}
    \label{fig:ggist}
    \end{figure}

    \begin{xample}\label{ex:GBmethod}
    We carry out the algorithm $\ggist$ for $F=3r_1^2+r_2^2+2r_1r_2$
    and $\bfmu=(2,1)$ as follows.
    \benum[Step 1]
    \item Construct $\calB=
        \{z_1-(2r_1+r_2),z_2-(r_1^2+2r_1r_2),z_3-r_1^2r_2\}$.
    \item Compute the Gr\"{o}bner basis of $\calB$ with the lexicographical order
            $z_1\prec z_2\prec z_3\prec r_1\prec r_2$ to get
        \begin{align*}
        \calG&=\{4z_1^3z_3-z_1^2z_2^2
                        -18z_1z_2z_3+4z_2^3+27z_3^2,
                        2r_1z_2^3+4z_1^2z_2z_3-z_1z_2^3
                        -54r_1z_3^2 \\
                    &~~+36z_1z_3^2-15z_2^2z_3, 6r_1z_1z_3-2r_1z_2^2-4z_1^2z_3
                        +z_1z_2^2+3z_2z_3,
                        r_1z_1z_2-9r_1z_3\\
                    &~~+6z_1z_3-2z_2^2,2r_1z_1^2-6r_1z_2-z_1z_2+9z_3,
                       3r_1^2-2r_1z_1+z_2,-z_1+2r_1+r_2\}.
        \end{align*}
    \item Compute the normal form of $F$ relative to $\calG$
            to get $R=z_1^2-z_2$.
    \item Since $\deg(R,\bfr)=0$, the algorithm outputs $R=z_1^2-z_2$.
    \eenum

    \end{xample}

\sect{Computing Gists via Preprocessing Approach}
\label{sec:Pesmethod}
    In the previous section, we show how to compute $\bfmu$-gists
    using Gr\"{o}bner bases. This algorithm is quite slow
    when $\bfmu\neq(1,1\dd 1)$
    \revised{(see Table \ref{tab:experiment1}, Example F3)}.
    In the next two sections,
    we will introduce two methods
    based on an analysis of the following two $K$-vector spaces:
    \bitem
        \item  $K^\delta\sym[\bfx]$: the set of symmetric homogeneous
                polynomials of degree $\delta$ in $K[\bfx]$
        \item  $K^\delta_\bfmu[\bfr]$: the set of
                $\bfmu$-symmetric polynomials of degree $\delta$ in $K[\bfr]$
    \eitem
    The first method is based on preprocessing and reduction:
    we first compute a basis for $K^\delta_\bfmu[\bfr]$,
    and then use the basis to reduce $F(\bfr)$.
    The second method directly computes the $\bfmu$-gist
    of $F(\bfr)$ by solving linear equations.

\ssect{Structure of a $\bfmu$-Symmetric Polynomial Set}
	\!\!\!\!We first consider \!$K^\delta\sym[\bfx]$, the
	symmetric homogeneous polynomials of degree $\delta$.
	This is a $K$-vector space.
    By a \dt{weak partition} of an integer $k$, we mean
    \revised{
    a sequence $\bfalpha=(\alpha_1,\alpha_2\dd \alpha_k)$
    where $\sum_{i=1}^k\alpha_i=k$
    and $\alpha_1\ge \alpha_2\ge\cdots\ge\alpha_k\ge 0$.
	Thus, in contrast to an ordinary partition,
    a weak partition allows zero parts.
    Given $\delta$,
    if $\bfalpha$ is a weak partition of $\delta$
    }
    and no part $\alpha_i$ larger than $n$,
    we will write
		$$\bfalpha\vdash (\delta,n).$$
	Let
		$$\e_\bfalpha \as \prod_{i=1}^\delta \e_{\alpha_i}$$
	For instance if $\delta=4, n=2, \bfalpha=(2,1,1,0)$ then
	$\e_\bfalpha=\e_2\e_1\e_1\e_0 = \e_2\e_1^2$.
	
	Let $T(\bfx)$ denote the set of terms of $\bfx$,
	and $T^\delta(\bfx)$ denote those terms
	of degree $\delta$.  A typical element of $T^\delta(\bfx)$ is
	$\prod_{i=1}^n x_i^{d_i}$ where $d_1+\cdots+d_n=\delta$.
	We totally order the terms in $T^\delta(\bfx)$
	using the lexicographic ordering in which
		$x_1 \prec x_2\prec \cdots\prec x_n$.
	Given any $F\in K(\bfx)$, its \dt{support} is
	$\Supp(F)\ib T(\bfx)$ such that $F$ can be uniquely
	written as
		\beql{F} F=\sum_{p\in \Supp(F)} c(p)p\eeql
	where
	$c:\Supp(F)\to K\setminus\set{0}$ denote the coefficients of $F$.
	Let the \dt{leading term} $\lt(F)$ be
	equal to the $p\in \Supp(F)$ which is the largest
	under the lexicographic ordering.  For instance,
	$\Supp(\e_1)\,=\,\set{x_1\dd x_n}$ and $\lt(\e_1)\,=\,x_n$.  Also
	$\Supp(\e_1\e_2)
		\,=\,\{x_i x_j x_k: 1\le i\neq j\le n,$ $1\leq k\leq n\}$
	and $\lt(\e_1\e_2)= x_n^2x_{n-1}$.
	The coefficient of $\lt(F)$ in $F$ is
	the \dt{leading coefficient} of $F$, denoted by $\lco(F)$.
	Call $\lm(F)\as \lco(F)\lt(F)$
	the \dt{leading monomial} of $F$.
    This is well-known:
	
	\bprol{basisofV1}
	The set $\calB_1 \as \set{\e_\bfalpha: \bfalpha\vdash (\delta,n)}$
	is a $K$-basis for the vector space $K^\delta\sym[\bfx]$.
	\eprol

    \begin{xample}\label{eg:V1space}
    Let $n=4$ and $\delta=3$.
	Then
		$\calB_1 = \set{\e_1^3,\e_1\e_2, \e_3}$
	forms a basis of the $K$-vector space $K^\delta\sym[\bfx]$.
    \end{xample}

    \ignore{%
	\bpf
	Let $F\in K^\delta\sym[\bfx]$ be non-zero, written
	in the standard form \refeq{F} above.
	Let $p_1=\lt(F)=\bfx_{\beta}$ where $\beta\vdash (\delta)_n$
    and $(\delta)_n$ denotes a $n$-partition of $\delta$
    allowing the appearance of zero. The symmetry of $F$ implies that
    $p_1$ is of the form $x_1^{\beta(1)}\cdots x_n^{\beta(n)}$
    where $\beta(1)\geq\cdots\geq\beta(n)\geq0$.
    We first prove that $p_1\in\lt(\calB_1)$.

    Rewrite $p_1$ into the following form
    \begin{align*}
    p_1&=x_1^{\beta(1)}\cdots x_n^{\beta(n)}\\
       &=(x_1\cdots x_n)^{\beta(n)}\cdot
        (x_1\cdots x_{n-1})^{\beta(n-1)-\beta(n)}
        \cdots
        (x_1x_2)^{\beta(2)-\beta(3)}\cdot
        x_1^{\beta(1)-\beta(2)}\\
       &=\lt\left(\e_n^{\beta(n)}\right)\cdot
        \lt\left(\e_{n-1}^{\beta(n-1)-\beta(n)}\right)\cdots
        \lt\left(\e_2^{\beta(2)-\beta(3)}\right)\cdot
        \lt\left(\e_1^{\beta(1)-\beta(2)}\right)\cdot
        \lt\left(\e_0^{\beta(0)-\beta(1)}\right)\cdot
        \\
       &=\lt\left(\prod_{i=0}^n\e_i^{\beta(i)-\beta(i+1)}\right).
    \end{align*}
    where $\beta(0)\as\delta$ and $\beta(n+1)\as 0$.
    Construct
    \[
    \alpha=(\underbrace{n\dd n}_{\beta(n)-\beta(n+1)},
            \underbrace{n-1\dd n-1}_{\beta(n-1)-\beta(n)}\dd
            \underbrace{2\dd 2}_{\beta(2)-\beta(3)},
            \underbrace{1\dd 1}_{\beta(1)-\beta(2)},
            \underbrace{0\dd 0}_{\beta(0)-\beta(1)}).
    \]
    Then the length of $\alpha$ is equal to $\delta$ because
    \[
    \sum_{i=0}^n(\beta(i)-\beta(i+1))=\beta(0)-\beta(n+1)=\delta.
    \]
    The summation of elements in $\alpha$ is also equal to $\delta$
    because
    \begin{align*}
    \sum_{i=0}^ni(\beta(i)-\beta(i+1))
            &=\beta(n)+(n-1)\beta(n)+\sum_{i=1}^{n-1}i(\beta(i)-\beta(i+1))\\
	    &=\beta(n)+(n-1)\beta(n)+(n-1)(\beta(n-1)-\beta(n))+\sum_{i=1}^{n-2}i(\beta(i)-\beta(i+1))\\
            &=\beta(n)+(n-1)\beta(n-1)+\sum_{i=1}^{n-2}i(\beta(i)-\beta(i+1))\\
            &=\beta(n)+\beta(n-1)+(n-2)\beta(n-2)+\sum_{i=1}^{n-3}i(\beta(i)-\beta(i+1))\\
            &=\cdots=\beta(n)+\beta(n-1)+\cdots+\beta(1)=\delta.
    \end{align*}
    Thus $\alpha$ is a partition of $\delta$ with the length $\delta$
    such that $p_1=\lt(\e_{\alpha})$.

    Setting $G_1\as\e_{\alpha}$
    and $F_1 \as F - c(p_1)G_1$, we see $F_1$ is also symmetric.
    If $F_1$ is non-zero, then $\deg(F_1,\bfx)=\deg(F,\bfx)=\delta$.
    Continue to form $F_2 \as F_1 -  c(p_2)G_2$ for
	$p_2=\lt(F_1)$ and also $p_2=\lt(G_2)\in \lt(\calB_1)$.
	After $k\ge 1$ steps, we obtain
	$F_{k+1}=F_k- c(p_k)G_k = 0$ with empty support.  Then we conclude that
	$F= \sum_{i=1}^k c(p_i) G_k$.  This proves that $K^\delta\sym[\bfx]$
	is linearly generated by $\calB_1$.

    The linear independence of polynomials in $\calB_1$ is immediate from
	the fact that $\lt(G)\neq \lt(G')$ for $G,G'\in \calB_1$
    if $G\neq G'$.
	\epf
    }

    \ignore{%
    It is easy to see that $p_1=\lt(e_1)$
	for some $e_1\in B_1$.  Setting $f_1 \as f - c(p_1)e_1$, we see
	that $\Supp(f_1)=\Supp(f)\setminus \Supp(e_1)$.
	If $\Supp(f_1)$ is non-empty, we continue
	to form $f_2 \as f_1 -  c(p_2)e_2$ for
	$p_2=\lt(f_1)$ and also $p_2=\lt(e_2)\in B_1$.
	After $k\ge 1$ such steps, we obtain
	$f_{k+1}=f_k- c(p_k)e_k = 0$ with empty support.  Then we conclude that
	$f= \sum_{i=1}^k c(p_i) e_k$.  This proves that $K^\delta\sym[\bfx]$
	is linearly generated by $B_1$.
    }
	
	Now we consider the set $K^\delta_\bfmu[\bfr]$ comprising
	the $\bfmu$-symmetric functions of degree $\delta$.
	The map
		$$\sigma_\bfmu: K^\delta\sym[\bfx] \to K^\delta_\bfmu[\bfr]$$
	is an onto $K$-homomorphism.
    Note that $K^\delta_\bfmu[\bfr]$ is a vector space which is generated by
	the set
		$$\mathcal{\olB}_1\as \set{\ol{G}: G\in \calB_1}$$
	where $\ol{G}$ is a short hand for writing $\sigma_\bfmu(G)$.
	It follows that there is a maximal independent
	set $\calB_2\ib \mathcal{\olB}_1$
    that is a basis for $K^\delta_\bfmu[\bfr]$.
	The set $\calB_2$ may be a proper subset of $\mathcal{\olB}_1$,
    which is seen in this example:
	let $\bfmu=(2,2)$ and $\delta=3$.
    From Example \ref{eg:V1space}, we have
    $\calB_1=\set{\e_1^3, \e_1\e_2,\e_3}$.
    Then
	$$\mathcal{\olB}_1=\set{A:\ole_1^3, B:\ole_1\ole_2,
		C:\ole_3}.$$
	We can check that $\mathcal{\olB}_1$ is linearly dependent
	since $A+8C=4B$.  Furthermore, it is easy to verify that any
	$2$-subset of $\mathcal{\olB}_1$ forms a basis for $K^\delta_\bfmu[\bfr]$.
    In general, we have the following lemma.

    \bpro
	For all
	$\bfmu=(\mu_1,\mu_2)$, $\mathcal{\olB}_1=\{\ole_1^2, \ole_2\}$
	is a linearly independent set.
    \epro
	\bpf
	Assume there exist $k_1$ and $k_2$ such that
	\begin{equation}\label{eq:s1s2linearindependent}
	k_1\ole_1^2+k_2\ole_2=0.
	\end{equation}
	Let $\bfmu=(\mu_1 \dd \mu_m)$. Then
		\begin{equation}\label{eqs:s1s2expression}
        \ole_1=\sum_{i=1}^m{\mu_ir_i},\quad
		\ole_2=\sum_{i=1}^m{\mu_i \choose 2}
		    r_i^2+\sum_{i<j}{\mu_i\mu_jr_ir_j}
		\end{equation}
	The substitution of \eqref{eqs:s1s2expression}
    into \eqref{eq:s1s2linearindependent} leads to
		\[\sum_{i=1}^m\left[k_1\mu_i^2+k_2{\mu_i\choose 2}\right]
			r_i^2+(2k_1+k_2)\sum_{i<j}\mu_i\mu_jr_ir_j=0.\]
	Therefore,
		\[k_1\mu_i^2+k_2{\mu_i\choose 2}=(2k_1+k_2)\mu_i\mu_j=0,
			\quad\mbox{for}\quad i,j=1 \dd m
			\quad\mbox{where}\quad i<j.\]
	This system has a unique solution which is $k_1=k_2=0$.
    Thus it follows that $\ole_1^2$ and $\ole_2$ are
	linearly independent.
	\epf

    From the previous discussion, we saw that the dimension of
    $K^\delta_\bfmu[\bfr]$ may be smaller than that of $K^\delta\sym[\bfx]$.
    There are two special cases:
    when $\bfmu=(1,1\dd 1)$,
		$dim(K^\delta\sym[\bfx])=dim(K^\delta_\bfmu[\bfr])$;
    when $\bfmu=(n)$, $dim(K^\delta_\bfmu[\bfr])=1$.
    The following table shows the dimensions of
    	$K^\delta\sym[\bfx]$ and $K^\delta_\bfmu[\bfr]$
    for some cases. One can see that it is quite common
    to have a dimension drop from the specialization $\sigma_\bfmu$
    (these lower dimensions are underlined in the table).

	\bcen
	\begin{table}[hbt]
	{\tiny
	\begin{tabular}{l l | l c c||l l | l c c} \hline
	    $n$ & ~~~~$\bfmu$ & $\delta$ & $dim(K^\delta\sym[\bfx])$
	    	& $dim(K^\delta_\bfmu[\bfr])$ &
        $n$ & ~~~~$\bfmu$ & $\delta$ & $dim(K^\delta\sym[\bfx])$
	    	& $dim(K^\delta_\bfmu[\bfr])$
	    \\\hline\hline
	    $3$ & $(2,1)$       & $2$ & $2$ & $2$
	 &  $5$ & $(2,1,1,1)$   & $4$ & $5$ &  $5$\\
	        &               & $3$ & $3$ & $3$
	 &      &               & $5$ & $7$ &  $7$\\
	        &               & $4$ & $4$ & $4$
	 &      &               & $6$ & $10$ &  $10$\\\hline
	    $4$ & $(2,1,1)$     & $3$ & $3$ & $3$
	 &  $5$ & $(2,2,1)$     & $4$ & $5$ &  $5$\\
	        &               & $4$ & $5$ & $5$
	 &      &               & $5$ & $7$  &  $7$\\
	        &               & $5$ & $6$  & $6$
	 &      &               & $6$ & $10$ &  $10$\\\hline
	    $4$ & $(3,1)$       & $3$ & $3$  & $3$
	 &  $5$ & $(3,1,1)$     & $4$ & $5$  &  $5$\\
		&               & $4$ & $5$  & $\ul{4}$
	 &      &               & $5$ & $7$  &  $7$\\
	        &               & $5$ & $6$  & $\ul{5}$
	 &      &               & $6$ & $10$ &  $10$\\\hline
	    $4$ & $(2,2)$       & $3$ & $3$  & $\ul{2}$
	 &  $5$ & $(3,2)$       & $4$ & $5$  &  $\ul{4}$\\
		    &               & $4$ & $5$  & $\ul{3}$
	 &      &               & $5$ & $7$  &  $\ul{5}$\\
		    &               & $5$ & $6$  & $\ul{3}$
	 &      &               & $6$ & $10$ &  $\ul{6}$\\\hline
	        &               &     &      &
	 &  $5$ & $(4,1)$       & $4$ & $5$  &  $\ul{4}$\\
	        &               &     &      &
	 &      &               & $5$ & $7$  &  $\ul{5}$\\
	        &               &     &      &
	 &      &               & $6$ & $10$ &  $\ul{6}$\\\hline
		\end{tabular}
		}
		\caption{Dimensions of $K^\delta\sym[\bfx]$
			and $K^\delta_\bfmu[\bfr]$}
		\end{table}
		\ecen


\ssect{Reduction and Canonical Sequence}
    This subsection is devoted to generating the basis of
    the vector space $K^{\delta}_{\bfmu}[\bfr]$
    with which one could easily check
    whether a given polynomial is in this vector space
    or not. For this purpose, we introduce a reduction procedure
    and its applications. This yields a more efficient method
    to check for $\bfmu$-symmetry
    \revised{and to compute the gists in the affirmative case}.

	A set $\calB\ib K[\bfr]$ is \dt{linearly independent}
	if any non-trivial $K$-linear combination over $\calB$ is non-zero;
	otherwise, $\calB$ is \dt{linearly dependent}.
	We say $\calC=(C_1\dd C_\ell)$ is a \dt{canonical sequence}
    if the set $\set{C_1\dd C_\ell}$ is linearly independent
    and $\lt(C_i)\prec\lt(C_j)$ for all $i<j$.
    \revised{
    In this subsection, we work in the vector space $K^{\delta}[\bfr]$
    of all homogeneous polynomials of degree $\delta$ in $K[\bfr]$.
    }

	We will introduce the concept of reduction.
	As motivation, first express any non-zero polynomial $G$ as
	$G=\lm(G)+R$ where $R$ is the tail of $G$
	(i.e., remaining terms of $G$).
	In the terminology of term rewriting systems (e.g.,
	\cite{sculthorpe+2:rewrite:14} and \cite[Section 12.3.4]{yap:algebra:bk}),
	we then view $G$ as a rule for rewriting an arbitrary polynomial
	$F$ in which any occurrence of $\lt(G)$ in $\Supp(F)$ is removed
	by an operation of the form
	$F' \ass F- c\cdot G$, with $c\in K$ chosen
    to eliminate $\lt(G)$ from $\Supp(F')$.
    For instance, consider
		$F= \ul{r_2^2} + 2r_1r_2 - r_1^2$
	and
		$G=\ul{r_1r_2} + r_1^2 - r_2$
	where we have underlined the leading monomials
	of $F$ and $G$. Here we use the above convention that $r_1\prec r_2$.
	Then
		$F' = F-2G = \ul{r_2^2} - 3r_1^2 + 2r_2$.
	We say that $F$ has been reduced by $G$ to $F'=F-2G$.
	The $\Supp(F')$ no longer has $r_1r_2$, but has gained other
	terms which are smaller in the $\prec$-ordering.

    If $\lt(G) \notin \Supp(F)$,
    we say $F$ is \emph{reduced} relative to $G$.
    For a sequence $\calC$, if $F$ is reduced with relative to each
    $G\in \calC$, we say $F$ is \emph{reduced} relative to $\calC$.
    Then we have this basic property:

    \bprol{linear_independence}
    Let $F\neq 0$ and $\calC=(C_1\dd C_\ell)$ be a canonical sequence.
    If $F$ is reduced relative to $\calC$,
    then $\set{F,C_1\dd C_\ell}$ is linearly independent.
    \eprol

    \bpf
    By way of contradiction, assume $F$ is linearly dependent
    on $\calC$, say  $F=\sum_{i=1}^\ell k_iC_i$.
    This implies $\lt(F)=\lt(\sum_{i=1}^\ell k_iC_i)\preceq\lt(C_\ell)$.
    So there is a smallest $j\le\ell$
    such that $\lt(F)\preceq\lt(C_j)$.  Since $F$ is reduced
    relative to $\calC$, we have $\lt(F)\prec\lt(C_j)$.
    It is easy to see that this implies $k_j,k_{j+1}\dd k_\ell$ are all zero.
    It follows that $j\ge 2$ (otherwise $F=\sum_{i=1}^\ell k_iC_i=0$).
    Moreover, we have
    $\lt\Big(\sum_{i=1}^\ell k_iC_i\Big)\preceq \lt(C_{j-1})\prec\lt(F)$.
    This contradicts the assumption $\sum_{i=1}^\ell k_iC_i =F$.
    \epf

    We next introduce the \reduce\ subroutine in
    Figure \ref{fig:reduce}
    which takes
    an arbitrary polynomial $F\in K[\bfr]$
    and a canonical sequence $\calC$
    as input to produce a reduced polynomial
    relative to $\calC$.

    \begin{figure}[h]
        \Ldent
        \progb{\\
            $\reduce(F,\calC)$:
            \lline[0] \INPUt:~~ $F\in K[\bfr]$,
                 $\calC=(C_1\dd C_\ell)$ is canonical
                    and each $C_i\in K^\delta[\bfr]$
            \lline[0] \OUTPUt: $R$ such that
                $F=\sum_{i=1}^\ell c_iC_i+R$ with $c_i\in K$ and
            \lline[14] $R$ is reduced relative to $\calC$.
            \lline[8] Let $R\ass 0$, $i\ass \ell$
            \lline[8] While ($F\neq0$ and $i>0$)
            \lline[16]  $p\ass \lt(F)$
            \lline[16]  If $p\succ\lt(C_i)$ then
            \lline[24]      $R\ass R+\lco(F)\cdot p$; $F\ass F-\lco(F)\cdot p$
            \lline[16]  else
            \lline[24]          If $p=\lt(C_i)$ then
            \lline[32]      	    $F\ass F-\frac{\lco(F)}{\lco(C_i)}C_i$
            \lline[24]          $i\ass i-1$
            \lline[8]  Return $R+F$
			\Ldent
        }
        \caption{The \reduce\ algorithm.}
        \label{fig:reduce}
    \end{figure}


    \begin{xample}\label{ex:reduce}
    Consider $F=3r_1^2+4r_1r_2+r_2^2$ and $\bfmu=(2,1)$.
    Given a canonical sequence
    $\calC=(r_1^2+2r_1r_2, 2r_1^2+r_2^2)$
    with $r_1\prec r_2$,
    we proceed to compute the reduced polynomial of $F$
    relative to $\calC$
    using the above $\reduce$ algorithm.
    \benum[Step 1]
    \item \emph{Initialization}. Let $R=0$ and $i=2$.
    \item \emph{First iteration}. For $F\neq0$ and $i>0$, $p=\lt(F)=r_2^2$
        which is equal to $\lt(C_2)$. Thus $F$ is updated with
        $F-\frac{\lco(F)}{\lco(C_2)}C_2=r_1^2+4r_1r_2$
        and $i$ is updated with $i-1=1$.
    \item \emph{Second iteration}. For $F\neq0$ and $i>0$, $p=\lt(F)=r_1r_2$
        which is equal to $\lt(C_1)$. Thus $F$ is updated with
        $F-\frac{\lco(F)}{\lco(C_1)}C_1=-r_1^2$
        and $i$ is updated with $i-1=0$.
    \item \emph{Finalization}.  Since $i=0$, the iteration stops
        and the algorithm outputs $R+F=-r_1^2$.
    \eenum

    \end{xample}

    \ignore{%
    The next two propositions imply the
    correctness of the $\reduce$ subroutine.
    }%

    \bpro
    The algorithm $\reduce(F,\calC)$ halts and takes at most
    $\#\Supp(F)-1+\sum_{i=1}^\ell\#\Supp(C_i)$ loops.
    Moreover, this bound is tight in the worst case.
    \epro

    \bpf
    Let $F_1$ denote the input polynomial.
    The variable $F$ in the algorithm is initially equal to $F_1$.
    In general, let $F_j$ ($j=1,2,\ldots$)
    be the polynomial denoted by $F$ at the beginning of the
    $j$th iteration of the while-loop.  Thus $p_j=\lt(F_j)$
    is the term denoted by the variable $p$ in the $j$th iteration.
    Note that $F_j$ transforms to $F_{j+1}$ by losing its
    leading term $p_j$ or furthermore, if $i(j)$ is the current value
    of the variable $i$, and
    $p_j=\lt(C_{i(j)})$ where $C_{i(j)}\in\calC$,
    we also subtract the
    tail of $\frac{\lco(F_j)}{\lco(C_{i(j)})}\cdot C_{i(j)}$
    from $F_{j+1}$.
    Thus,
    $\Supp(F)\ib \Supp(F_1)\cup \Supp(\calC)$.
    Since $p_1\succ p_2\succ \cdots$ and
    $p_j\in \Supp(F_1)\cup \Supp(\calC)$, this proves that the
    algorithm halts after at most $\#\Supp(F_1)+\#\Supp(\calC)$ iterations.

    Let $L$ be the actual number of iterations.
    We now give a refined argument to show that
    $L\le \#\Supp(F)-1+\#\Supp(\calC)$, i.e., we can improve the
    previous upper bound on $L$ by one.
    Note that we exit the while-loop when $F=0$ or $i=0$ holds.
    There are two cases.

    CASE 1: $F=0$ \emph{and} $i=0$ both hold.
    This implies that in the previous
    iteration, $p_L = \lt(C_1)$, and $i$ was decremented from $1$ to $0$.
    Since $p_L$ came from $\#\Supp(F_1)$ or $\#\Supp(C_2\dd C_\ell)$,
    this implies
	$$L\le \#(\Supp(F_1)\cup \Supp(\calC))
		\le \#\Supp(F_1)-1+\#\Supp(\calC).$$

    CASE 2: $F\neq 0$ or $i>0$.  Each iteration can be ``charged''
    	to an element of $\#(\Supp(F_1)\cup \Supp(\calC))$.
	If $i>0$, then some elements in $\Supp(C_1)$ are not charged.
	If $F\neq 0$, then $\Supp(F)\ib \Supp(F_1)\cup\Supp(\calC)$
	also implies that some elements of
		$\Supp(F_1)\cup \Supp(\calC)$ are not charged.
	Thus CASE 2 implies
		$$L\le \#\Supp(F_1)-1+\#\Supp(\calC).$$
	This proves our claimed upper bound on $L$.

    To prove that this bound is tight,
    let $F_1=p_1+q_1+\cdots+q_s$ and $\calC=(p_1\dd p_\ell)$
    with the term ordering
    $p_1\prec\cdots\prec p_\ell\prec q_1\prec\cdots \prec q_s$.
    In the first $s$ loops,
    since $\lt(F_1)\succ p_\ell$,
    $i$ is unchanged and
    $q_1\dd q_s$ are removed from $F$.
    In the next $\ell-1$ loops,
    since $\lt(F_1)=p_1\prec p_2\prec\cdots\prec p_\ell$,
    $F$ is unchanged and $i$ will drop to $1$.
    In the last loop,
    since $\lt(F_1)=p_1=\lt(C_1)$,
    $F$ will be reduced relative to $C_1$ to $0$.
    So the total number of loops is
    $s+\ell=\#\Supp(F_1)-1+\sum_{i=1}^\ell \#\Supp(C_i)$.
    \epf

	\bpro
    (Correctness) The \reduce\ subroutine is correct.
    \epro

    \bpf
    Correctness of the output $R_*$ in the \reduce\ subroutine
    amounts to two assertions.
    \\ (A1) The output $R_*$ is reduced relative to $\calC$.
    \\ (A2) $F_1-R_*$ is a linear combination of the polynomials in $\calC$
    	where $F_1$ is the input polynomial.
    \\ To prove these assertions,
    assume that the while-loop terminates after the $L$-th iteration.
    Also let $F_j$, $R_j$ and $i_j$ denote
    the values of the variables $F$, $R$ and $i$ at the start of
    the $j$th iteration (for $j=1\dd L,L+1$).
    Thus, $F_1$ is the input polynomial, $R_1=0$ and $i_1=\ell$.
    Assertion (A2) follows from the fact that in
    each iteration, the value of $F+R$ does not change
    or it changes by a scalar multiple of some $C_i\in\calC$.
    To see Assertion (A1), we use induction
    on $j$ to conclude that $F_j$ is reduced with respect
    to $\calC_j \as (C_{1+i_j}, C_{2+i_j}\dd C_{\ell})$,
    and $R_j$ is reduced with respect to $\calC$.
    Finally, the output $R_*$ is equal to $R_{L+1}+F_{L+1}$,
    At termination, there are two cases: either $F_{L+1}=0$
    (so $R_*=R_{L+1}$)
    or $i_{L+1}=0$ (so $R_*=R_{L+1}+F_{L+1}$).
    In the first case,
    Assertion (A1) holds because
    $R_*=R_{L+1}$ and $R_{L+1}$ is reduced w.r.t.~$\calC$.
    In the second case, Assertion (A1) holds because
    $F_{L+1}$ is reduced w.r.t. $\calC_{L+1}=\calC$.
    \epf

    \bpro
	If $\calC=(C_1\dd C_\ell)$ is canonical,
    then $\reduce(F,\calC)\!=\!0$ iff $\set{F,C_1\dd C_\ell}$
    is linearly dependent.
	\epro
	\bpf
	One direction is immediate:
	$\reduce(F,\calC)=0$ implies that $F$ is a linear combination
	of the elements of $\calC$.
    Conversely, if $\reduce(F,\calC)=F'\neq 0$,
    then $\set{F',C_1\dd C_\ell}$ is linearly independent
    by \refPro{linear_independence}.
    Moreover, $F'=F-\sum_{i=1}^\ell k'_iC_i$ for some $k'_1\dd k'_\ell$.
    By way of contradiction, assume that $\set{F,,C_1\dd C_\ell}$
    is linearly dependent, i.e.,
    $F=\sum_{i=1}^\ell k_iC_i$ for some $k_1\dd, k_\ell$.
    It follow that $F'=\sum_{i=1}^\ell (k_i-k'_i)C_i$,
    contradicting the linear independence of $\set{F',C_1\dd C_\ell}$.
	\epf
	
	This gives rise to the $\canonize$ algorithm in Figure
	\ref{fig:canonize} to construct a canonical sequence.

    \begin{figure}[h]
    \Ldent
    \progb{\\
            $\canonize(\calB)$:
	    \lline[-4] \INPUt:~~ $\calB=(B_1\dd B_\ell)$ where \revised{$B_i\in K^\delta[\bfr]$}.
            \lline[-4] \OUTPUt: a canonical $\calC$ whose linear span
                                satisfies span$(\calB)=$span$(\calC)$
            \lline[0] Let $\calC\ass ()$ (empty sequence)
            \lline[0] For $i=1$ to $\ell$
            \lline[8]  $B \ass \reduce(B_i,\calC)$
            \lline[8]  If $B\neq 0$ then
            \lline[16]  $\calC\ass \Insert(B,\calC)$
            \lline[0] Return $\calC$
			\Ldent
        }
        \caption{The $\canonize$ algorithm.}
        \label{fig:canonize}
    \end{figure}

    We view the sequence $\calC=(C_1\dd C_m)$
    as a sorted list of polynomials,
    with $\lt(C_i)\prec\lt(C_{i+1})$.  Thus
    $\Insert(B,\calC)$ which inserts $B$ into $\calC$,
    can be implemented in $O(\log m)$ time
    with suitable data structures.  The overall complexity is
    $O(\ell+ m\log m)$ where $m$ is the length of the output $\calC$.
    Alternatively, we could initialize
    the input $\calB$ as a priority queue can pop the
    polynomial $B\in \calB$ with the largest $\lt(B)$.  This design yields
    a complexity of $O(\ell\log \ell)$ which is inferior when $\ell\gg m$.

    \begin{xample}\label{ex:canonize}
    Consider a polynomial set $\calB=\set{4r_1^2+4r_1r_2+r_2^2,r_1^2+2r_1r_2}$.
    We proceed to compute a canonical sequence from $\calB$ relative to
	$r_1\prec r_2$ using the $\canonize$ algorithm.
    \benum[Step 1]
    \item \emph{Initialization}. Let $\mathcal{C}=()$.
    \item \emph{First iteration}. Let $B=r_1^2+2r_1r_2$.
        Note that $\mathcal{C}=()$. Thus $B'=\reduce(B, \mathcal{C})$ $=B$
        and $\mathcal{C}$ is updated with $(r_1^2+2r_1r_2)$.
    \item \emph{Second iteration}. Let $B=4r_1^2+4r_1r_2+r_2^2$.
        Then carry out the reduction of $B$ relative to $\mathcal{C}$ and we get
        $B'=\reduce(B, \mathcal{C})=2r_1^2+r_2^2$.
        After inserting $B'$ into $\mathcal{C}$,
        $\mathcal{C}$ is updated with $(r_1^2+2r_1r_2, 2r_1^2+r_2^2)$.
    \item \emph{Finalization}.  Now the iteration stops
        and the algorithm outputs $\mathcal{C}=(r_1^2+2r_1r_2, 2r_1^2+r_2^2)$.
    \eenum

    \end{xample}

    \ignore{%
	Chee to Jing: I do not see why canonize is correct.  Don't
	you need to assume that "$pop(\calB)$" removes the polynomial
	in $\calB$ with the largest $\lt(B_i)$???
	So the first thing you do is to construct a priority queue $Q$
	containing $\calB$?
    }
	
    The termination of $\canonize(\calB)$ is immediate from
    the termination of $\reduce(F,$ $\calC)$.
	The correctness of the output of $\canonize(\calB)$ comes from two
	facts:
    the returned $\mathcal{C}$ is clearly canonical.
    It is also maximal because any element $B\in \calB$
    that does not contribute to $\mathcal{C}$
    is clearly dependent on $\mathcal{C}$.

    It should be pointed out that by tracking the ``quotients" of $F$
    relative to $\calC$ in the \reduce\ algorithm
    and integrating the information into the \canonize\ algorithm,
    we can derive the relationship between
    $\calB=\{\ole_{\bfalpha}: \bfalpha\vdash(\delta,n)\}$
    and $\calC=\canonize(\calB)$ and write polynomials in $\calC$
    as linear combinations
    of polynomials in $\calB$.
    By ``quotients'', we mean the coefficients $c_i$'s in the expression
    $F=\sum_{i=1}^\ell c_iC_i+R$.
    \revised{
    When the quotient information is required,
    we use algorithms $\reduce(F,\calC,`q`)$ and $\canonize(\calB,`Q`)$
    where $q$ and $Q$ represents the quotient (column) vector and quotient matrix, respectively.
    More explicitly,
    \[F=\calC\cdot q+R\quad\mbox{and}\quad\calC=\calB\cdot Q\]
    where $\calB$ and $\calC$ are viewed as row vectors.
    These notations will be used in the \crgist\ algorithm
    in Figure \ref{fig:crgist} of the following subsection.
    }

\ignore{
	\blem \label{lem:termination-reduce}
	   $\reduce(f,B)$ terminates in no more than $\#B$ steps.
	\elem

	\bpf
	Without losing generality, we assume that polynomials in $B$
	are ordered by the leading terms. Let $k=\#B$.
	Then $B=\{b_1,\dd,b_k\}$ where $\lt(b_i)\prec \lt(b_{i+1})$.
	We prove the lemma by inductions for $k$.

	When $k=1$, the termination is immediate. Assume that
	the termination is true for $k>1$. In what follows,
	we prove that it is also true for $k+1$.

	\begin{itemize}
	\item Case 1: $\lt(f)\neq\lt(b_1)$.
	Choose $b \in B\setminus\{b_1\}$ such that $\lt(b)=\lt(f)$.
	Then the result of the subtraction step $f - \frac{\lco(f)}{\lco(b)} b$
		does not have
	any new term which precedes $\lt(b_2)$, including $\lt(b_1)$.
	Thus $b_1$ will never be used in $reduce(f,B)$.
	In other words, $reduce(f,B)=reduce(f, B\setminus\{b_1\})$
	and the latter terminates which is guaranteed by the induction.

	\item Case 2: $\lt(f)=\lt(b_1)$. In the first step of $reduce(f,
	    B)$ , we choose $b=b_1$,
	then $\lt(f - \frac{\lco(f)}{\lco(b)} b)\neq \lt(b_1)$;
	thus the remaining part of $reduce(f, B)$ is $reduce(f -
		\frac{\lco(f)}{\lco(b)} b, B\setminus\{b_1\})$,
	which always terminates by induction.
	\end{itemize}
	\epf
}

\ssect{Computing $\mathbf\mu$-gist via reduction}
    In this subsection, we use $\reduce$ and $\canonize$
    algorithms to construct the $\crgist$ algorithm
    for computing the $\bfmu$-gist of a polynomial.

    \begin{figure}[h]
    \centering
		\Ldent\progb{\lline[-8]
            		$\crgist(F,\bfmu)$:
            \lline[-4] \INPUt:~~
				$F\in K^\delta[\bfr]$, $\bfmu=(\mu_1\dd\mu_m)$
            \lline[-4] \OUTPUt: the $\bfmu$-gist of $F$ if $F$ is
	    	$\bfmu$-symmetric; otherwise
	    \lline[15] return ``$F$ is not $\bfmu$-symmetric".
            \lline[0] $\delta\ass\deg(F,\bfr)$
            \lline[0] $n\ass\sum_{i=1}^m\mu_i$
            \lline[0] $\calB\ass(\ole_{\bfalpha}: \bfalpha\vdash(\delta,n))$
            \lline[0] $Z\ass(z_{\bfalpha}: \bfalpha\vdash(\delta,n))$
            \lline[0] $\mathcal{C},Q\ass \canonize(\calB,`Q`)$
            \lline[0] $R,q\ass\reduce(F,\mathcal{C},`q`)$
            \lline[0] If $R=0$ then
            \lline[8]  Return $Z\cdot Q\cdot q$
            \lline[0] Return \emph{``$F$ is not $\bfmu$-symmetric''}
			\Ldent
        }
        \caption{The $\crgist$ algorithm.}
        \label{fig:crgist}
    \end{figure}
    \begin{xample}\label{ex:checkmusymmetry}
    Consider the polynomial $F=3r_1^2+4r_1r_2+r_2^2$ and $\bfmu=(2,1)$
    as in Example \ref{ex:GBmethod}. In what follows,
    we check whether $F$ is $\bfmu$-symmetric or not and compute its $\bfmu$-gist
    in the affirmative case.

    \benum [Step 1]
    \item Let $\delta=\deg(F,\bfr)=2$ and $n=\sum_{i=1}^m\mu_i=3$.
    \item From $\delta$ and $n$, construct $\calB=\set{(2r_1+r_2)^2,r_1^2+2r_1r_2}$ and $Z=(z_1^2,z_2)$.
    \item Compute a canonical $\mathcal{C}$ from $\calB$ and its quotient $Q$ relative to $\calB$. Then we get
        $\mathcal{C}=\canonize(\calB)$ $=(r_1^2+2r_1r_2, 2r_1^2+r_2^2)$ and
        $Q=\left(
        \begin{array}{cc}
        0&1\\
        1&-2
        \end{array}\right)$.
        The detailed computation can be found in Example \ref{ex:canonize}.
    \item Compute $R=\reduce(F,\mathcal{C})$ and the quotient $q$. By the result of Example
	\ref{ex:reduce},
        $R=-r_1^2\neq0$ and $q=(2,1)^T$. Thus the output is ``No",
        which means that $F$ is not $\bfmu$-symmetric.
    \eenum

    If we replace $F$ with $F=3r_1^2+2r_1r_2+r_2^2$, then after carrying out
    the same procedure as above, we will get $R=0$ and $q=(1,1)$,
    which means $F$ is $\bfmu$-symmetric and its $\bfmu$-gist is
    \[\ooF=(z_1^2,z_2)\cdot Q\cdot q^T=z_1^2-z_2.\]
    \end{xample}

    Since termination of the algorithm $\crgist$
    is immediate from that of $\canonize$ and $\reduce$,
    we only show its correctness. Assume $\deg(F,\bfr)=\delta$.
    Recall that $F\in K[\bfr]$ is $\bfmu$-symmetric
    iff there exists a homogeneous symmetric
    polynomial $\whF\in K[\bfx]$ of degree $\delta$
    such that $\sigma_{\bfmu}(\whF)=F(\bfr)$.
    By \refPro{basisofV1}, $\whF$ is symmetric
    and with degree $\delta$ iff $\whF\in K^\delta\sym[\bfx]$.
    Thus $F=\sigma_{\bfmu}(\whF)\in K^\delta_\bfmu[\bfr]$
    where $K^\delta_\bfmu[\bfr]$ is a $K$-vector space
    with the basis generated by
    $\calB=\{\ole_\bfalpha: \bfalpha\vdash(\delta,n)\}$.
    If $\calC= \canonize(\calB)$,
    then $\calC$ is the basis we want to obtain.
    Therefore, if $F$ is $\bfmu$-symmetric iff $\reduce(F,\calC)=0$.
    \revised{
    When $F$ is $\bfmu$-symmetric, $F=\calC\cdot q=\calB\cdot Q\cdot q$.
    By the definition of $\bfmu$-gist,
    $\ooF=(z_{\bfalpha}: \bfalpha\vdash(\delta,n))\cdot Q\cdot q$.
    }

\ssect{Exponential lower bound for nondeterministic reduction}
	In this subsection, we consider an alternative reduction process
	where each reduction step is non-deterministic.
    We prove that this version can be exponential in the worst case.
	
	For any term $p$, let $\coef(F,p)$ denote
	the coefficient of $p$ in $F$.  If $p\notin\Supp(F)$,
	then $\coef(F,p)=0$.  For any polynomial $C$, define
		$$\reduceStep(F,C)\ass
			F-\frac{\coef(F,\lt(C))}{\lco(C)} C.$$
	We call $\reduceStep(F,C)$ a \dt{$C$-reduction step}
    or a $\calC$-reduction step in case $C\in \calC$.
    We see that $\reduceStep(F,C)=F$ iff $\lt(C)$ does not occur in $F$.
    We say the reduction is improper in this case.

	Let $\nreduce(F,\calC)$ denote the subroutine
    that repeatedly transforms $F$ by applying proper
    $\calC$-reduction steps to $F$ until no more
	more change is possible. It returns the final value of $F$.
	We call this the \dt{nondeterministic reduction} of $F$.

	\bprol{nred_equal_red}
	For any linearly independent set $\calC$, we have  	
		$$\nreduce(F,\calC)=\reduce(F,\calC).$$
	Then $\nreduce(F,\calC)$ has $\leq 2^\ell$ $\calC$-reduction steps
	where $\ell=|\calC|$.  Moreover, $2^\ell$ steps may be needed.
	\eprol

	\bpf
    Let $R_1=\nreduce(F,\calC)$ and $R_2=\reduce(F,\calC)$.
    Then there exists $k_1\dd k_\ell$ and $k_1'\dd k_\ell'$ such that
    $$F=\sum_{i=1}^\ell k_iC_i+R_1=\sum_{i=1}^\ell k_i'C_i+R_2.$$
    It is immediate that
    	$$R_1-R_2=\sum_{i=1}^\ell (k_i-k_i')C_i.$$
    If $R_1\neq R_2$, there exists $i$ such that
    $k_i\neq k_i'$ and $k_j=k_j'~(j=1\dd i-1)$.
    Then $\lt(R_1-R_2)=\lt(C_i)$.
    This implies that $\lt(C_i)\in \Supp(R_1)$ or $\lt(C_i)\in\Supp(R_2)$.
    Hence $R_1$ or $R_2$ is not reduced relative to $\calC$.
    This contradicts with the output requirements of $\reduce$ or $\nreduce$.

    Let us define $a_\ell$ to be the longest $\calC$-derivation
    for any $\calC$ with $\ell$ elements.
    	CLAIM A: $a_\ell\le 2^\ell-1$.
	Let $\calC_\ell=(C_1\dd C_\ell)$ be any canonical sequence
	with $\ell$ elements.  Let
	\beql{f0}
	F_0\to F_1\to\cdots \to F_N
	\eeql
	be any $\calC_\ell$-derivation.  We must prove that $N\le 2^\ell-1$
	by induction of $\ell$.   Clearly, if $\ell=1$, then
	$a_1\le 1=2^1-1$.  Next, inductively
	assume that $a_{\ell-1}\le 2^{\ell-1}-1$.
	Suppose there does not exist an $i<N$ such that
	$F_i\to F_{i+1}$ is a $C_\ell$-reduction step.
	In that case, \refeq{f0} is a $\calC_{\ell-1}$-derivation.
	By induction hypothesis, $N\le 2^{\ell-1}-1<2^\ell-1$, as
	claimed.  Otherwise, we may choose the smallest $i$
	such that $F_i\to F_{i+1}$ is a $C_\ell$-reduction
	step.  Note that this implies that $\lt(C_\ell)$
	does not appear in the support of $F_j$ for all $j\ge i+1$.
	In other words, $F_0\to\cdots\to F_i$ and
	$F_{i+1}\to\cdots \to F_N$ are both $\calC_{\ell-1}$-derivations.
	By induction hypothesis, both these lengths are at most
	$a_{\ell-1} \le 2^{\ell-1}-1$.  Thus the length of \refeq{f0}
	is at most $2a_{\ell-1}+1 \le 2^{\ell}-1$.  Thus CLAIM A is proved.

	The last assertion of our proposition amounts
	to CLAIM B: $a_\ell\ge 2^\ell-1$.
	To show this claim, let $\calC_\ell=(C_1\dd C_\ell)$ as before.
	But we now
	choose $C_i \as \sum_{j=1}^i p_j$ where $p_j$'s are terms satisfying
    	$p_j\prec p_{j+1}$.  Let us write
    		$$F \xrightarrow[ k ]{ \calC } G$$
	to mean that there is a
	$\calC$-derivation of length $k$ from $F$ to $G$.
   	Our claim follows if we show that
    		$$C_\ell \xrightarrow[ 2^\ell-1 ]{ \calC_\ell } 0.$$
	The basis is obvious:
    		$C_1 \xrightarrow[ 1 ]{ \calC_l } 0.$
	Inductively, assume that
		\beql{bi1}
		C_{\ell-1} \xrightarrow[ 2^{\ell-1}-1 ]{ \calC_{\ell-1} } 0.
		\eeql
	The inductive assumption implies
		$$C_{\ell} = p_{\ell}+C_{\ell-1}
			\xrightarrow[ 2^{\ell-1}-1 ]{ \calC_{\ell-1}}
			p_{\ell}.$$
	Next, in one step, we have
		$p_{\ell} \xrightarrow[ 1 ]{ \calC_{\ell} } -C_{\ell-1}$
	and, again from the induction hypothesis,
		$$-C_{\ell-1}
			\xrightarrow[ 2^{\ell-1}-1 ]{ \calC_{\ell-1} } 0.$$
	Concatenating these 3 derivations, shows that
	$C_\ell\xrightarrow[ 2^\ell-1] {\calC_\ell} 0$.
	This proves CLAIM B.

    \ignore{
    To see the upper bound of $2^\ell-1$ is tight, let $\calB=(B_1\dd B_\ell)$
	where $B_1=p_1$ and $B_{i+1}= B_i + p_{i+1}$ ($i=1\dd \ell$).
	Here, $p_1\dd p_\ell$ are terms with $p_i\prec p_{i+1}$.
    Consider $F=B_\ell$. The worst reduction is executed
    in the above recursive way. For example,
    when $\ell=3$, $\calB=(p_1,p_1+p_2,p_1+p_2+p_3)$ and $F=p_1+p_2+p_3$.
    Then the worst reduction is carried out as follows:
    \[
    F=p_1+p_2+p_3\xrightarrow{B_1}p_2+p_3\xrightarrow{B_2}-p_1+p_3
        \xrightarrow{B_1}p_3\xrightarrow{B_3}-p_1-p_2
        \xrightarrow{B_1}-p_2\xrightarrow{B_2}p_1\xrightarrow{B_1}0.
    \]
    The total number of steps is $7$, which is equal to $2^3-1$.
	}
	\epf

\sect{Computing Gists via Solving Linear Equations}
\label{sec:gist}
    In this section, we introduce a direct method
    to compute gist of $F(\bfr)$ without preprocessing.
    Such methods depend on the choice of basis for
    $K^\delta\sym[\bfr]$. Our default basis is elementary symmetric polynomials.

    Our algorithm that takes as input
    $F\in K[\bfr]$ and $\bfmu$,
    and either outputs the $\bfmu$-gist $\ooF$ of $F$ or
    detects that $F$ is not $\bfmu$-symmetric.
    The idea is this: $F$ is $\bfmu$-symmetric iff $\ooF$ exists.
    The existence of $\ooF$ is equivalent to
    the existence of a solution to a linear system of equations.
    More precisely, there is an polynomial identity of the form
		$\ooF(\ole_1\dd \ole_n)=F.$
	To turn this identity into a system of linear equations, we
	first construct a polynomial
		$$G(\bfk;\bfz)\in K[\bfk][\bfz]$$
	in $\bfz$ with indeterminate coefficients in $\bfk$,
	with homogeneous weighted degree $\delta$ in $\bfz$
	(see Section \ref{ssec:lift+ideal} for definition of weighted degree).
    Here $\delta$ is the degree of $F$.
    Each term is of weighted degree $\delta$ and has the form
		$$\bfz_\bfalpha \as \prod_{i=1}^\delta z_{\alpha_i}$$
	where $\bfalpha=(\alpha_1\dd \alpha_\delta)$ is
	a weak partition of $\delta$ with parts at most $n$,
	i.e., $\bfalpha\vdash (\delta,n)$.
	\revised{
    Then $G(\bfk;\bfz)$ can be written as
		$$G(\bfk;\bfz)\as \sum_{\bfalpha\vdash (\delta,n)}
				k_\bfalpha \bfz_\bfalpha=T^\delta_n(\bfz)\cdot\bfk$$
	where
    $T^\delta_n(\bfz) \as
    (\bfz_\bfalpha: \bfalpha\vdash (\delta,n))$ and
	$\bfk \as (k_\bfalpha: \bfalpha\vdash (\delta,n))^T$
    viewed as a column vector are indeterminates.}
	Next, we plug in $\ole_i$'s for the $z_i$'s to get
		$$H(\bfk;\bfr)\as G(\bfk;\ole_1\dd \ole_n)$$
	viewed as a polynomial in
		$K[\bfk][\bfr].$
	We then set up the equation
		\beql{hbfk}
		H(\bfk;\bfr) = F(\bfr)
		\eeql
	to solve for the values of $\bfk$.
	Note that total degree of $G$ in $\bfk$ is $1$, i.e.,
	$\deg(G,\bfk)=1$.  Therefore, $\deg(H,\bfk)=1$.
	Thus \refeq{hbfk} amounts to solving a linear system
	of equations in $\bfk$.

	To illustrate this process, consider the
    polynomial $F=3r_1^2+2r_1r_2+r_2^2$ and $\bfmu=(2,1)$.
    \benum[Step 1:]
    \item Assign $\delta= \deg(F,\bfr)=2$ and $n= \sum_{i=1}^m\mu_i=3$.
    \item \revised{
            Since the weak partitions of $2$ with parts at most $3$ are
    		$(1,1)$ and $(2,0)$,}
		 the terms of weighted degree $2$ are $z_1^2$ and $z_2$.
    \item Construct the polynomial
	$G(\bfk;\bfz)\as k_1z_1^2+k_2z_2$
	where $\bfk=(k_1,k_2)$ are the indeterminate coefficients.
    \item Using $\ole_1=2r_1+r_2, \ole_2=r_1^2+2r_1r_2$,
	construct the polynomial
	$$H(\bfk;\bfr)\as G(\bfk;\ole_1\dd\ole_n)=
		(4k_1+k_2)r_1^2 + (4k_1+2k_2)r_1r_2 + k_1r_2^2.$$
    \item Extract the coefficient vector $\coeffs(H,\bfr)$
		of $H(\bfk;\bfr)$ viewed as a polynomial in $\bfr$.
		The entries of this vector are linear in $\bfk$.
		Thus $H= \coeffs(H,\bfr)\cdot T^\delta(\bfr)$
		where $T^\delta(\bfr)$ is the vector of
		all terms of $T(\bfr)$ of degree $\delta$.
    \item Extract the coefficient vector $\coeffs(F,\bfr)$
		of $F(\bfr)$.  This vector is a constant $(3,2,1)^T$
        \revised{where $T^2(r_1,r_2)=(r_1^2,r_1r_2,r_2^2)$}.
    \item The last two steps enables the construction of a system
	of linear equations, $A\bfk =\bfb$:
		\beqarrays
		H(\bfk;\bfr) &=& F(\bfr)\\
		(4k_1+k_2)r_1^2 + (4k_1+2k_2)r_1r_2 + k_1r_2^2 &=&
    			3r_1^2+2r_1r_2+r_2^2\\
		\coeffs(H,\bfr) &=& \coeffs(F,\bfr)\\
		\mmat{
		    4 & 1\\
		    4 & 2\\
		    1 & 0}\cdot\mmat{k_1\\k_2} &=& \mmat{3\\ 2\\ 1} \\
		    A \cdot \bfk &=& \bfb
		    \eeqarrays
	where the last equation is the linear system to be solved for
	$\bfk=\mmat{k_1\\k_2}$.
    \item
	If $A \bfk=\bfb$ has no solutions, we conclude that $F$ is
	not $\bfmu$-symmetric.  Otherwise, choose any solution
	for $\bfk$ and plugging into $G(\bfk;\bfr)$,
	we obtain a gist of $F(\bfr)$.
    \revised{
    Here $\bfk=\mmat{1\\-1}$ is a solution
    and thus the input polynomial is $(2,1)$-symmetric with gist $z_1^2-z_2$.}
	Note that there may be multiple solutions for $\bfk$
	because of the presence of $\bfmu$-constraints.
    \eenum

	We now summarize the above procedure as the \lsgist\ algorithm:

	\begin{figure}[h]
    \Ldent
	\progb{\\
	    $\lsgist(F,\bfmu)$:
	    \lline[0] \INPUt:~~ $F\in K^\delta[\bfr]$ and
	    		$\bfmu=(\mu_1\dd\mu_m)$
	    \lline[0] \OUTPUt: the $\bfmu$-gist of $F$ if $F$ is
	    	$\bfmu$-symmetric; otherwise
	    \lline[15] return ``$F$ is not $\bfmu$-symmetric".
	    \lline[8] $\delta\ass\deg(F,\bfr)$; $n\ass\sum_{i=1}^m\mu_i$
	    \lline[8] $G\ass
	    	\sum_{\bfalpha\vdash (\delta,n)} k_{\bfalpha}z_{\bfalpha}$
	    \lline[8] $H\ass G(\bfk;\ole_1\dd\ole_n)$
	    \lline[8] Extract $\coeffs(H,\bfr)$ and $\coeffs(F,\bfr)$.
	    \lline[8] Find a solution $\bfk=\bfk_0$ of the linear system
	    \lline[20] $\coeffs(H,\bfr)=\coeffs(F,\bfr)$.
	    \lline[8] If $\bfk_0$ is nondefined
	    \lline[16]  Return \emph{``$F$ is not $\bfmu$-symmetric''}
	    \lline[8] Else
	    \lline[16] Return $H(\bfk_0;\bfr)$
		\Ldent
		}
		\caption{The $\lsgist$ algorithm.}
        \label{fig:lsgist}
    \end{figure}

    The correctness of the algorithm $\lsgist$ lies in
    the fact that $F$ is $\bfmu$-symmetric iff $F\in K^\delta_\bfmu[\bfr]$
    which is generated by $\{\ole_{\bfalpha}: \bfalpha\vdash(\delta,n)\}$.

\sect{Gists Relative to Other Bases of $K^\delta\sym[\bfx]$}
\label{sec:OtherBases}
    In this section, We briefly sketch
    how to extend the above methods to computing gists relative to
    other bases of $K^\delta\sym[\bfx]$.

    The set $K\sym[\bfx]$ of symmetric functions
    can be viewed as a $K$-algebra
    generated by some finite set $\calG$.
    The following are three well-known choices of $\calG$ with $n$ elements each:
    \bitem
\item (Elementary symmetric polynomials)
        $\calG_e\as \set{\e_1\dd\e_n}$ where $e_i$ is the $i$-th elementary
	symmetric function of $\bfx$.
    \item (Power-sum symmetric polynomials)
        $\calG_p\as \set{p_1\dd p_n}$ where $p_i=x_1^i+\cdots+x_n^i$.
    \item (Complete homogeneous symmetric polynomials)
        $\calG_c\as \set{c_1\dd c_n}$
        where $c_i$ is the sum of all distinct monomials of
        degree $i$ in the variables $x_1\dd x_n$.
    \eitem

    For each $\delta\ge 1$, the vector space $K\sym^\delta[\bfx]$ of
    symmetric polynomials of degree $\delta$ has
    a basis $\calB^\delta$ that corresponds to
    a given generator set $\calG$.
    \ignore{%
    To describe $\calB$, we use index set
    $I^\delta_n\as\set{\bfalpha: \bfalpha\vdash(\delta,n)}$,
    comprising those weak partitions $\bfalpha$ of $\delta$
    with \revise{$\delta$ parts and} no part larger than $n$.
    }%
    The following are bases of $K\sym^\delta[\bfx]$:
    \bitem
\item (\dt{$e$-basis})
    	$\calB_e^\delta\as \set{e_{\bfalpha}: \bfalpha\vdash(\delta,n)}$
	where $e_{\bfalpha}=\prod_{i=1}^\delta e_{\alpha_i}$
	and $\bfalpha=(\alpha_1\dd\alpha_\delta)$;
    \item (\dt{$p$-basis})
    	$\calB_p^\delta\as \set{p_{\bfalpha}: \bfalpha\vdash(\delta,n)}$
	where $p_{\bfalpha}=\prod_{i=1}^\delta p_{\alpha_i}$;
    \item (\dt{$c$-basis})
    	$\calB_c^\delta\as \set{c_{\bfalpha}: \bfalpha\vdash(\delta,n)}$
	where $c_{\bfalpha}=\prod_{i=1}^\delta c_{\alpha_i}$.
    \eitem
    But $K\sym^{\delta}[\bfx]$ can also
    be generated with monomial symmetric polynomials.
    In this case, \revised{we use $\bfalpha\vdash (\delta)_n$
    to denote $\bfalpha=(\alpha_1\dd \alpha_n)$
    which is a weak partition of $\delta$ with}
    exactly $n$ parts: $\alpha_1\ge\cdots\ge\alpha_n\ge 0$.
    We also write $\bfx^\bfalpha$ for the product
    $\prod_{i=1}^n x_i^{\alpha_i}$.  This yields yet another basis
    for $K\sym^\delta[\bfx]$:
    \bitem
    \item
    (\dt{$m$-basis})
    $\calB_m^\delta\as \set{m_{\bfalpha}: \bfalpha\vdash(\delta)_n}$
    where
    $m_{\bfalpha}=\sum_{\bfbeta} \bfx^\bfbeta$ where
    $\bfbeta$ ranges over all permutations of $\bfalpha$ which
    are distinct.
    \eitem
    For instance, if $\bfalpha=(2,0,0)$ then
    $\bfbeta$ ranges over the set $\set{(2,0,0), (0,2,0),(0,0,2)}$
    and $m_\bfalpha=x_1^2+x_2^2+x_3^2$.

    So far, this paper has focused on the $e$-basis.
    But concepts and algorithms
    \revised{relative to the choice of this basis
    (e.g., the $\bfmu$-gist and $\ggist$)
    can be reformulated using the other bases.
    In each algorithm, there are two parameters,
    i.e., the generator polynomials (e.g., $e_i$ and $\ole_i$)
    and the index set (e.g., $\bfalpha\vdash (\delta,n)$).
    }
    When using $p$-basis or $c$-basis,
    we only need to replace $\ole_i$
    used by the algorithms $\ggist$, $\crgist$ and $\lsgist$
    by
    $\olp_i\as\sigma_{\bfmu}(p_i)$
    or $\olc_i\as\sigma_{\bfmu}(c_i)$,
    respectively;
    when using the $m$-basis, the index set
    $\bfalpha\vdash (\delta,n)$
    should be replaced by $\bfalpha\vdash (\delta)_n$ and
    $\ole_i$ should be replaced by
    $\olm_\bfalpha\as\sigma_{\bfmu}(m_\bfalpha)$.
    The relative performance of the algorithms using different bases
    will be evaluated in Section \ref{sec:experiments}.

    \ignore{%
    Chee asked:
    why does the above not mentioned the m-basis?
    I think we should say explicitly that
      \\ \lsgist\ and \crgist\ are the basic algorithms.
      \\ We have three variants for each of these basic algorithms.
      (Name them).
      \\ Another glaring gap is this:
      why don't we introduce variants of
      the \ggist\ algorithm?  It seems just as simple
      to introduce the ideals
      	$I^p_\bfmu \as \set{z_i-\olp_i:i=1\dd n}$
	or
      	$I^c_\bfmu \as \set{z_i-\olc_i:i=1\dd n}$.
      }

\sect{Experiments}\label{sec:experiments}
    In this section, we report some experimental results to show the
    effectiveness and efficiency of the two approaches presented
    in this paper. These experiments were performed using Maple
    on a Windows laptop with an Intel(R) Core(TM) i7-7660U CPU in 2.50GHz
    and 8GB RAM.

    In \refTab{experiment1}, we
    compare the performance of the three algorithms
    described in this paper for checking the $\bfmu$-symmetry
    of polynomials: \ggist, \lsgist\ and \crgist.
    We use a test suite of $12$ polynomials of degrees ranging from $6$--$20$
    (see \refTab{experiment1}), with corresponding $\bfmu$
    with $n=|\bfmu|$ ranging from $4$--$6$.
    These polynomials are either $\dplus$ polynomials
    or subdiscriminants, or some perturbations (to create
    non-$\bfmu$-symmetric polynomials).

    \begin{table}[htb]
    {\tiny %
    \centering \caption{\label{tab:experiment1}
	Comparing the performance of \ggist, \lsgist\ and \crgist.
        Computing the $\bfmu$-gist of $F$ of degree $\delta$.
	Here $n=\sum_{i=1}^m\mu_i$,
	$\canonize$ is a preprocessing step in
	\crgist\
    and total=$\canonize$ time + $\reduce$ time.
    }
    \begin{tabular}{|l|c|c|c|c|c|c|c|c|c|c|c|}\hline
    \multirow{3}{*}{~F} & \multirow{3}{*}{$\delta$} & \multirow{3}{*}{$\bfmu$}
    & \multirow{3}{*}{$n$} & \multirow{3}{*}{\!\!Y/N\!\!}  &{\ggist}  &  {\lsgist}
    & {speedup}  & \multicolumn{3}{c|}{\crgist}     & {speedup}\\
    \cline{9-11} & & & & & Time & Time & (\ggist/& $\canonize$
                        & $\reduce$ & total & (\ggist/\\
    & & & & &(sec)&(sec)&\lsgist\ )& (sec) & (sec) & (sec) & \crgist) \\
    \hline\hline
    \input{./Performance-EA}
    \end{tabular}
    }
    \vspace{-0.5em}
    \end{table}

    From Table \ref{tab:experiment1}, it is clear that
    \lsgist\ is significantly faster than \ggist.
    There is one anomaly in the table: for the polynomial $F_3$,
    \ggist\ is 5 times faster than \lsgist.  This is when
    $\bfmu$ is $(1\dd 1)$, which indicates that the ideal
    $\calI_\bfmu=\left<v_1\dd v_n\right>$
    has a symmetric structure in $\bfr$.
    We believe it is because the Gr\"{o}bner basis of $\calI$ can be computed
    very efficiently for certain types of structures.

    \begin{table}[h]
    {\tiny
    \centering \caption{\label{tab:experiment2} Timing for
	computing the gists of $\bfmu$-symmetric polynomials
	\revised{
    with Gr\"{o}bner basis method using different bases
    (i.e., \eGgist, \pGgist\ and \cGgist),
    with \canonize+\reduce\
    using different bases
    (i.e., \eCRgist, \pCRgist, \cCRgist\ and \mCRgist)
    and with linear system solving using different bases
	(i.e., \eLSgist, \pLSgist, \cLSgist\ and \mLSgist).}
    The most efficient method for each case is marked with *
    next to the running time.
    }

    \begin{tabular}{|c|c|c|c|c|c|c|c|c|c|c|c|}\hline
    \multirow{3}{*}{\!\!\!F\!\!\!}
    & \multicolumn{3}{c|}{\!\!\!Gr\"{o}bner basis method\!\!\!}
    & \multicolumn{4}{c|}{\!\!\!\canonize+\reduce\!\!\!}
    & \multicolumn{4}{c|}{\!\!\!Linear system solving\!\!\!}
    \\\cline{2-12}
    &\!\!\!\eGgist\!\!\!&\!\!\!\pGgist\!\!\!&\!\!\!\cGgist\!\!\!
    &\!\!\!\eCRgist\!\!\!&\!\!\!\pCRgist\!\!\!&\!\!\!\cCRgist\!\!\!&\!\!\!\mCRgist\!\!\!
    &\!\!\!\eLSgist\!\!\!&\!\!\!\pLSgist\!\!\!&\!\!\!\cLSgist\!\!\!&\!\!\!\mLSgist\!\!\!\\
    & (sec) & (sec) & (sec)
    & (sec) & (sec) & (sec) & (sec)
    & (sec) & (sec) & (sec) & (sec) \\
\hline
    \input{./Performance-FP}
    \end{tabular}
    }
    \end{table}

    From Table \ref{tab:experiment2}, we observe that
    the algorithm \mLSgist\ is more efficient
    than the other algorithms in general
    because it doesn't require polynomial expansion
    and thus can save a lot of time, especially
    when $\delta$ is big (see F3, F4 and F7).
    The algorithm \pLSgist\ also behaves well
    because power-sum symmetric polynomials have fewer terms than
    elementary symmetric polynomials and complete homogeneous symmetric
    ones and this property may help save time during polynomial expansion.
    Overall, algorithms based on \canonize+\reduce\
    are not as competitive as those based on linear system solving
    because the preprocessing procedure \canonize\ charges more time
    in order to generate a canonical sequence.

    \begin{table}[h]
    \scriptsize
    \centering \caption{\label{tab:experiment4} Timing for
	computing the gists using \ggist, \canonize+\reduce\ and
    linear system solving with $e$-basis
    when $\bfmu$ and $\delta$ are fixed. Here $\bfmu=(2,2,1)$ and $\delta=10$)}
    \begin{tabular}{|c|c|c|c|c|c|c|c|c|c|}\hline
       &        &\multicolumn{2}{c|}{\eGgist}    &\multicolumn{2}{c|}{\eCRgist}&\eLSgist\\
    \cline{3-6}
    F~ & Y/N    & $GroebnerBasis$ & $NormalForm$ & \canonize\ & \reduce\      & Time \\
       &        & Time (sec)      &  Time (sec)  & Time (sec) & Time (sec)    & (sec)\\
    \hline
    {F10} & Y &  \multirow{4}{*}{37.2}  &  0.188 &  \multirow{4}{*}{0.063}  &  0.000 &  0.063\\
    \cline{1-2}\cline{4-4}\cline{6-7}
    {F11} & Y &                           &  0.203 &                          &  0.016 &  0.016\\
    \cline{1-2}\cline{4-4}\cline{6-7}
    {F12} & Y &                           &  0.203 &                          &  0.000 &  0.046\\
    \cline{1-2}\cline{4-4}\cline{6-7}
    {F13} & N &                           &  0.344 &                          &  0.000 &  0.062\\
    \hline
    \multicolumn{2}{|c|}{Total time} &\multicolumn{2}{c|}{38.1} &  \multicolumn{2}{c|}{0.079\rstar} & 0.187\\
    \hline
    \end{tabular}
    \end{table}

    However, from Table \ref{tab:experiment4},
    we see that for fixed $\bfmu$ and $\delta$,
    once we have computed the canonical set in the preprocessing step,
    the time cost for $\reduce$ is small.
    \revised{
    Therefore, when evaluating the total time for several examples
    sharing the same canonical set,
    the algorithm based on \canonize+\reduce\
    can be superior to the linear solving method.
    }
    Although the Gr\"{obner} basis method also contains a preprocessing procedure,
    the time cost for computing normal forms is quite expensive
    and thus it is not as competitive as algorithms based on
    \canonize+\reduce\ and linear system solving.
    Furthermore, for algorithms using $p$-basis,
    the algorithm \pLSgist\ shows higher efficiency than \pGgist\ and \pCRgist,
    especially for big $\delta$ and $n$ (See F3, F4 and F7).
    This could be attributed to the small number of terms
    in the generator polynomials.
    In contrast, for $e$-basis and $c$-basis,
    the algorithms \eCRgist\ and \cCRgist\ prevail over
    \eLSgist\ and \cLSgist.
    The possible reason might be that many terms will get
    canceled when computing a canonical sequence.

\sect{Conclusion}\label{sec:conclusion}
    We have introduced the concept of $\bfmu$-symmetric polynomial
    which generalizes the classical symmetric polynomial.
    Such $\bfmu$-symmetric functions of the roots of a polynomial
    can be written as a rational function in its coefficients.
    Our original motivation was to study a conjecture that
    a certain polynomial $\dplus(\bfmu)$ is $\bfmu$-symmetric.
    \revised{
    In order to explore such properties for different $\bfmu$'s
    and other root functions,
    we introduce three algorithms to compute the
    $\bfmu$-gist of a polynomial (or detect that no such gists
    exist).
    With the help of these algorithms, we verified the $\bfmu$-symmetry of $\dplus$
    for many specific cases.
    In a companion paper \cite{yang-yap:dplus:20},
    we will prove the $\bfmu$-symmetry conjecture on $\dplus$
    and show its application in the complexity analysis of root clustering.
    }

\ignore{%
\section*{Appendix}

Assume that $K$ is a field over $\mathbb{Q}$ finitely generated by finitely many elements, and $\overline{K}$ stands for the algebraic closure of $K$. Denote by $\Omega_n(K)$ the set of irreducible polynomials of degree $n$ in $K[x]$ whose Galois groups are $S_n$. Set
$$
   \Gamma=\left\{\bfgamma\in \overline{K}^n \,\,\left | \prod_{i=1}^n(x-\gamma_i)\in \Omega_n(K) \right.\right\}.
$$
where $\bfgamma=(\gamma_1\dd\gamma_n)$.

\blem\label{LM:dense}
Suppose that $P\in K[\bfx]$. If $P(\bfgamma)=0$ for all $\bfgamma\in \Gamma$ then $P=0$.
\elem

\bpf
Suppose on the contrary that there is a nonzero polynomial $P\in K[\bfx]$ satisfying that $P(\bfgamma)=0$ for all $\bfgamma\in \Gamma$.
Let $e_i$ be the $i$-th elementary symmetric polynomial in $x_1,\cdots,x_n$.
Let
$$I=\langle P,\,\, y_1-e_1(\bfx), \cdots,y_n-e_n(\bfx)\rangle$$
and
$$J=\langle y_1-e_1(\bfx), \cdots,y_n-e_n(\bfx)\rangle \rangle $$
where $\langle * \rangle$ denotes the ideal in $K[\bfx,\bfy]$ generated by $*$. Then $J$ is a prime ideal of dimension $n$. It is easy to see that $P\notin J$. Hence $\dim(I)<\dim(J)=n$, which implies that there is a nonzero polynomial $h(\bfy)\in K[\bfy]\cap I$. Now denote
$$
   \Sigma=\left\{(\bfgamma,e_1(\bfgamma), \cdots,e_n(\bfgamma))\in \overline{K}^{2n} \,\,\left | \bfgamma\in \Gamma\right.\right\}.
$$
Then $\Sigma\subset \mathbb{V}(I)$ where $\mathbb{V}(I)$ denotes the set of zeroes of $I$ in $\overline{K}^{2n}$. Proposition 2 on the page 123 of \cite{serre:lectures:bk} and Remarks on the page 129 of the same reference imply that there is an irreducible polynomial $f(x)=x^n-c_{n-1}x^{n-1}+\cdots+(-1)^nc_0\in \Omega_n(K)$ such that $h(c_{n-1},\cdots,c_0)\neq 0$.  On the other hand, let $\gamma_1,\cdots,\gamma_n$ be all roots of $f(x)=0$ in $\overline{K}$. Then $(\gamma_1,\cdots,\gamma_n, c_{n-1},\cdots,c_0)\in \Sigma$. Hence $h(c_{n-1},\cdots,c_0)=0$, a contradiction.
\epf

\bprol{symmetry}
Suppose that $P\in K[\bfx]$. If $P(\bfgamma)\in K$ for all $\bfgamma\in \Gamma$ then $P$ is symmetric.
\eprol

\bpf
For each $\sigma\in S_n$, denote $Q_\sigma=\sigma(P)-P$ where $\sigma(P)=P\left(x_{\sigma(1)}, \cdots,x_{\sigma(n)}\right)$. As $P(\bfgamma)\in K$, due to the Galois theory,
$$\sigma(P)(\bfgamma)=\sigma(P(\bfgamma))=P(\bfgamma)$$
for all $\sigma$ in the Galois group of $K(\bfgamma)$ over $K$, where $K(\bfgamma)$ denotes the field over $K$ generated by the entries of $\bfgamma$. For each $\bfgamma\in \Gamma$, we have that the Galois group of $K(\bfgamma)$ over $K$ is $S_n$. Hence $Q_\sigma(\bfgamma)=0$ for all $\sigma\in S_n$ and all $\bfgamma\in \Gamma$. By Lemma~\ref{LM:dense}, $Q_\sigma=0$ for all $\sigma\in S_n$. Therefore $P$ is symmetric.
\epf
}%

\ignore{
\bibliographystyle{abbrv}
\bibliography{test,st,yap,exact,geo,alge,math,com,rob,cad,algo,visual,gis,quantum,mesh,tnt,fluid,p}

%% file: Performance-EA.tex
{F1} &  12 & [1, 1, 1, 1] &  4 & Y  &  0.453 &  0.235 &    1.9 &  0.094 &  0.000 &  0.094 &    4.8 \\
 \hline
{F2} &   8 & [2, 1, 1] &  4 & Y  &  0.328 &  0.015 &   21.9 &  0.016 &  0.015 &  0.031 &   10.6 \\
 \hline
{F3} &  20 & [1, 1, 1, 1, 1] &  5 & Y  & \cored{34.1} & \cored{188} &    \cored{0.2} &  3.77 &  0.031 &  3.80 &    9.0 \\
 \hline
{F4} &  15 & [2, 1, 1, 1] &  5 & Y  & $>$600 &  1.88 &  $>$320 &  0.391 &  0.015 &  0.406 & $>$1478 \\
 \hline
{F4x} &   6 & [2, 1, 1, 1] &  5 & N  & $>$600 &  0.015 & $>$4$\times\!10^4$ &  0.000 &  0.016 &  0.016 & $>$3.7$\times\!10^4$ \\
 \hline
{F5} &   6 & [2, 2, 1] &  5 & Y  & 68.0 &  0.032 & 2126 &  0.000 &  0.000 &  0.000 &    Inf \\
 \hline
{F5x} &   6 & [2, 2, 1] &  5 & N  &  0.078 &  0.000 &    Inf &  0.000 &  0.016 &  0.016 &    4.9 \\
 \hline
{F6} &  10 & [2, 2, 1] &  5 & Y  &  0.438 &  0.078 &    5.6 &  0.031 &  0.000 &  0.031 &   14.1 \\
 \hline
{F6x} &  10 & [2, 2, 1] &  5 & N  &  0.406 &  0.047 &    8.6 &  0.031 &  0.016 &  0.047 &    8.6 \\
 \hline
{F7} &  18 & [3, 1, 1, 1] &  6 & Y  & $>$600 &  9.00 &   $>$66.7 &  3.39 &  0.063 &  3.45 &  $>$174 \\
 \hline
{F8} &  12 & [3, 2, 1] &  6 & Y  & $>$600 &  0.360 & $>$1667 &  0.187 &  0.000 &  0.187 & $>$3210 \\
 \hline
{F9} &   6 & [2, 2, 2] &  6 & Y  &  8.73 &  0.000 &    Inf &  0.000 &  0.000 &  0.000 &    Inf \\
 \hline

%% file: Performance-FP.tex
{F1} &  0.219 &  0.344 &  0.187 &  0.063 &  0.187 &  0.078 &  0.094 &  0.297 &  0.094 &  0.500 &  0.031\rstar  \\
 \hline
{F2} &   0.328 & 387 & 565 &  0.015 &  1.00 &  0.032 &  0.015 &  0.016 &  0.000\rstar &  0.031 &  0.000\rstar  \\
 \hline
{F3} &   20.9 & 61.2 & 60.3 &  3.19 & 313 & 14.7 &  2.17 & 79.3 &  3.11 & 2109 &  0.391\rstar  \\
 \hline
{F4} &  $>$3000 & $>$3000 & $>$3000 &  0.422 &  5.06 &  1.58 &  0.437 &  0.907 &  0.281 &  5.19 &  0.110\rstar  \\
 \hline
{F4x} &  $>$3000 & $>$3000 & $>$3000 &  0.000\rstar &  0.015 &  0.016 &  0.016 &  0.015 &  0.016 &  0.031 &  0.000\rstar  \\
 \hline
{F5} &  41.2 & $>$3000 & $>$3000 &  0.016 &  0.016 &  0.015 &  0.016 &  0.000\rstar &  0.015 &  0.016 &  0.000\rstar  \\
 \hline
{F5x} &  0.047 & $>$3000 & $>$3000 &  0.015 &  0.016 &  0.015 &  0.016 &  0.000\rstar &  0.016 &  0.015 &  0.000\rstar  \\
 \hline
{F6} & 0.234 & $>$3000 & $>$3000 &  0.047 &  0.094 &  0.093 &  0.063 &  0.031\rstar &  0.031\rstar &  0.063 &  0.032  \\
 \hline
{F6x} &  0.281 & $>$3000 & $>$3000 &  0.047 &  0.094 &  0.093 &  0.063 &  0.031\rstar &  0.031\rstar &  0.047 &  0.031\rstar  \\
 \hline
{F7} & $>$3000 & $>$3000 & $>$3000 &  3.50 & 63.3 & 29.8 &  3.50 &  5.63 &  1.70 & 49.2 &  1.52\rstar  \\
 \hline
{F8} & $>$3000 & $>$3000 & $>$3000 &  0.234 &  0.641 &  0.656 &  0.438 &  0.156 &  0.109\rstar &  0.250 &  0.157  \\
 \hline
{F9} &  6.17 & $>$3000 & $>$3000 &  0.016 &  0.000\rstar &  0.031 &  0.016 &  0.000\rstar &  0.015 &  0.031 &  0.016  \\
 \hline